\documentclass[journal,twoside,web]{ieeecolor}
\usepackage{generic}
\usepackage{cite}
\usepackage{amssymb,amsfonts}
\usepackage{amsmath}
\usepackage{autonum}
\usepackage{url}
\usepackage{graphicx}
\usepackage{algorithm,algorithmic}
\PassOptionsToPackage{hyphens}{url}
\usepackage{textcomp}
\usepackage{booktabs}
\usepackage[switch]{lineno}
\usepackage{acronym}
\usepackage[english]{babel}

\usepackage{changes}
\usepackage{paralist}
\usepackage{mathrsfs}  

\usepackage{mathtools}

\usepackage{tikz}
\usetikzlibrary{arrows, arrows.meta}
\usetikzlibrary{automata,positioning}
\usetikzlibrary{shapes.geometric}

\def\BibTeX{{\rm B\kern-.05em{\sc i\kern-.025em b}\kern-.08em
    T\kern-.1667em\lower.7ex\hbox{E}\kern-.125emX}}
  \newcommand{\nat}{\mathbf{N}}

  \newcommand{\eg}{\textit{e.g.}}

\newcommand{\dist}{\mathcal{D}}

  \newcommand{\supp}{\mathsf{Supp}}

\newcommand{\truev}{\mathsf{true}}

\newcommand{\Eventually}{\diamondsuit \, }

\newcommand{\reals}{\mathbb{R}}

\newtheorem{theorem}{Theorem} 
\newtheorem{definition}{Definition}
\newtheorem{proposition}{Proposition}

\newtheorem{problem}{Problem}

\newtheorem{example}{Example}
\newtheorem{remark}{Remark}

\newtheorem{lemma}{Lemma}
\newtheorem{assumption}{Assumption}

\acrodef{mdp}[MDP]{Markov decision process} 
\acrodef{lmdp}[LMDP]{labeled Markov decision process} 

\acrodef{asw}[ASW]{Almost-Sure Winning}
\acrodef{ltlf}[LTL$_f$]{Linear Temporal Logic over Finite Traces}
\acrodef{ltl}[LTL]{linear temporal logic}
\acrodef{scltl}[co-safe LTL]{syntactically co-safe Linear Temporal Logic}
\acrodef{dfa}[DFA]{deterministic finite automaton}
\acrodef{des}[DES]{discrete event system}

 \newcommand{\init}{q_0}

 \DeclareMathOperator*{\optmax}{\mathrm{maximize}}

 \DeclareMathOperator*{\optst}{\mathrm{subject\ to}}

\newcommand{\expect}{\mathbb{E}}
\newcommand{\obs}{E}

\newcommand{\calS}{\mathcal{S}}
\newcommand{\calA}{\mathcal{A}}
\newcommand{\calAP}{\mathcal{AP}}

\acrodef{hmm}[HMM]{hidden Markov model}
\acrodef{fsc}[FSC]{finite state controller}
\acrodef{pomdp}[POMDP]{partially observable Markov decision process} 

\newcommand{\matO}{{\mathbf{O}}}
\newcommand{\matT}{{\mathbf{T}}}
\newcommand{\matA}{{\mathbf{A}}}

\definecolor{darkgreen}{rgb}{0,0.5,0}

\begin{document}
\title{ Policy Gradient Methods for Information-Theoretic Opacity in Markov Decision Processes }

\author{Chongyang Shi$^{1}$, \IEEEmembership{Graduate Student Member, IEEE}, Sumukha Udupa$^{1}$, \IEEEmembership{Graduate Student Member, IEEE},\\ Michael R. Dorothy$^{2}$, \IEEEmembership{Member, IEEE}, Shuo Han$^{3}$, \IEEEmembership{Member, IEEE} and Jie Fu$^{1}$, \IEEEmembership{Member, IEEE}
\thanks{$^{1}$Chongyang Shi, Sumukha Udupa, and Jie Fu are with the Department of Electrical and Computer Engineering, University of Florida, Gainesville, FL, USA.
        {\tt\small \{c.shi, sudupa, fujie\}@ufl.edu}.}%
\thanks{$^{2}$Michael R. Dorothy is with DEVCOM Army Research Laboratory, Adelphi, MD, USA.
        {\tt\small michael.r.dorothy.civ@army.mil}.}%
\thanks{$^{3}$Shuo Han is with the Department of Electrical and Computer Engineering, University of Illinois Chicago, Chicago, IL, USA.
        {\tt\small hanshuo@uic.edu}.}%
}

\maketitle

\begin{abstract}
Opacity, or non-interference, is a property ensuring that an external observer cannot infer confidential information (the ``secret'') from system observations. We introduce an information-theoretic measure of opacity, which quantifies information leakage using the conditional entropy of the secret given the observer’s partial observations in a system modeled as a  Markov decision process (MDP). Our objective is to find a control policy that maximizes opacity while satisfying task performance constraints, assuming that an informed observer is aware of the control policy and system dynamics. Specifically, we consider a class of opacity called state-based opacity, where the secret is a propositional formula about the past or current state of the system, and a special case of state-based opacity called language-based opacity, where the secret is defined by a temporal logic formula (LTL) or a regular language recognized by a finite-state automaton. First, we prove that finite-memory policies can outperform Markov policies in optimizing information-theoretic opacity. Second, we develop an algorithm to compute a maximally opaque Markov policy using a primal-dual gradient-based algorithm, and prove its convergence. Since opacity cannot be expressed as a cumulative cost, we develop a novel method to compute the gradient of conditional entropy with respect to policy parameters using observable operators in hidden Markov models.  The experimental results validate the effectiveness and optimality of our proposed methods.
\end{abstract}

\begin{IEEEkeywords}
Information theory and control, Markov processes, optimization, stochastic systems.
\end{IEEEkeywords}

\section{INTRODUCTION}
Opacity generalizes and unifies various notions of confidentiality, including secrecy, anonymity, and privacy. 
It focuses on concealing sensitive aspects of system behavior, thereby enhancing security and privacy measures in various applications. For example, opacity has wide applications in the analysis of cyber-physical systems~\cite{lin2011opacity, yin2019approximate, Parv2013Privacy}, and state estimation~\cite{SHU20083054, molloy2023smoother}. Different types of opacity have been studied, including \emph{state-based}, 
which requires the secret behavior of the system (i.e., the membership of its initial/current/past state to the set) to remain opaque~\cite{saboori2007notions,saboori2013current,HanX2023Scai,saboori_verification_2013,Saboori2009Kstep,yinInfinitestepOpacityKstep2019}; and \emph{language-based} opacity, wherein the secret is defined by a language or a temporal logic formula~\cite{dubreil2008opacity,lin2011opacity}. This concept was initially introduced in seminal works such as~\cite{Mazar2004UsingUF, bryansOpacityGeneralisedTransition2006}, in which a language-based opacity is enforced qualitatively if the observer cannot distinguish a trajectory satisfying the secret from one violating the secret, given his partial observation.  


Qualitative opacity, commonly studied in discrete event systems (DESs), characterizes whether a system can entirely conceal its secret behavior from an external observer. Formally, a system is qualitatively opaque if, for every trajectory whose final state lies in the secret set $G$, there exists at least one observation-equivalent trajectory whose final state does not belong to $G$. In other words, based solely on the observations, the observer cannot determine with certainty whether the secret has occurred. 
In stochastic systems, however, qualitative opacity is often inadequate for confidentiality protection, since it only ensures the possibility of concealing secrets without quantifying the likelihood of doing so. Even if a system is qualitatively opaque, probabilistic inference may still lead to significant information leakage. For example, consider a system that can generate two trajectories producing the same observation sequence: one that includes the secret state with probability $0.99$, and another that excludes it with probability $0.01$. Under the qualitative definition, the system remains opaque, yet an observer can infer that the secret state is reached with high probability ($0.99$), revealing substantial information about the secret. 

With this motivation, we introduce the notion of \emph{information-theoretic opacity}. 
Building on Shannon's work on secrecy~\cite{shannonCommunicationTheorySecrecy1949} and the concept of guesswork~\cite{khouzani2017leakage}, we propose to measure opacity using the conditional entropy of the secret given the observations. 
The idea comes from Fano's inequality in information theory~\cite{Wiley2005information}:
\[
P_E \ge \frac{H(X|Y) - 1}{\log(|\mathcal{X}|-1)},
\]
where $P_E$ is the probability of incorrect guess, $H(X|Y)$ is the conditional entropy of a random variable $X$ given $Y$; $|\mathcal{X}|$ is the size of the support of $X$. Maximizing the conditional entropy $H(X|Y)$ effectively increases the lower bound of the probability of guess error. Consequently, this enhances the confidentiality of the secret. 

We study two notions of information-theoretic state-based opacity: last-state opacity and initial-state opacity. The distinction lies in the placement of the secret state: in the former, the secret is defined as the final state of a run, whereas in the latter, it is defined as the initial state. Building on these concepts, we then introduce information-theoretic language-based opacity, which generalizes opacity to entire sequences of system behaviors.
Based on the relation between temporal logic and automaton~\cite{vardi2005automata}, the secret in language-based opacity is defined to be the automaton state that captures the progress in satisfying a linear temporal logic formula. This definition generalizes existing language-based opacity in supervisory control, where the secret is whether the formula is satisfied or not~\cite{wintenberg2022general}.

Given the entropy-based measure of opacity, we investigate a class of constrained opacity-enforcement problems, where the objective is to maximize opacity with respect to a given secret in \iac{mdp}, 
while ensuring that task performance constraints are satisfied. We consider an \textit{informed observer} who has complete knowledge of both the system model and the control policy. Through a simple example, we demonstrate that a finite-memory policy can outperform a Markov policy in optimizing opacity. Then, we develop gradient-based algorithms for computing an opacity-maximizing 
Markov policy. We propose a novel method based on observable operators~\cite{jaeger2000observableoperator} to compute or approximate the gradient of conditional entropy with respect to the control policy parameters. Leveraging this gradient estimation, we employ a primal–dual gradient-based optimization technique to compute a locally optimal policy that maximizes opacity while satisfying a constraint on the total return.  Moreover, we establish the differentiability of conditional entropy with respect to policy parameters, which guarantees the convergence of the gradient algorithm. Finally, we demonstrate the effectiveness and near-optimality of the proposed approach through simulations in both a grid world environment and a graph \ac{mdp} example.

\textbf{Related work:}
In supervisory control, 
different types of qualitative opacity and corresponding methods to enforce opacity have been investigated, including state-based opacity \cite{saboori2007notions,saboori2013current,HanX2023Scai,saboori_verification_2013,Saboori2009Kstep,yinInfinitestepOpacityKstep2019}, language-based opacity \cite{dubreil2008opacity,lin2011opacity,shi2023synthesis}, 
and model-based opacity \cite{keroglouProbabilisticSystemOpacity2016}, the last of which aims to prevent the observer from discerning the true model of the system among multiple candidates. 
While qualitative opacity assesses whether a system is opaque or not, a quantitative analysis of opacity is crucial for quantifying the degree of opacity within a system. In \cite{berardQuantifyingOpacity2015, berardProbabilisticOpacityMarkov2015},  quantitative opacity is defined as the probability of generating an opaque run in a stochastic system.
Recent work \cite{Yin2021approximate} introduces \textit{approximate opacity}, which is achieved by constructing a symbolic abstraction -- a deterministic transition system -- of a continuous-state control system using a simulation relation, and
thereby extending opacity analysis to continuous systems. 

In \cite{Chen2023symbolic}, the authors present a symbolic differential privacy (DP) mechanism that protects a sensitive word by generating a random word that
is likely nearby. Although both opacity and DP are concerned with hiding information,  
they differ significantly in both the objective and the technical approach: DP is to ensure that individual data points cannot be easily inferred from an aggregated output based on data, even against   adversaries with any side information, while opacity is to ensure that the observer is maximally uncertain about some secret with partial observations of the system.
Rather than restricting 
inference about individual data points, opacity quantifies the residual uncertainty of the adversary regarding the secret. Technically, differential privacy relies on randomized perturbations on the aggregated output
, whereas opacity is enforced by designing system behaviors   that strategically obscure the secret from the observer.

Our definition of information-theoretic opacity is quantitative in nature but differs from differential privacy~\cite{mirInformationTheoreticFoundationsDifferential2013}. Differential privacy introduces perturbation-based mechanisms to reduce the mutual information between the output and any individual datum to a certain level, typically by adding controlled noise (e.g., Laplace or Gaussian). In contrast, opacity in control systems does not rely on perturbations; instead, it seeks to design a control policy that limits the information an observer can infer about a system’s secret from its observable behavior. 
While both frameworks share the goal of concealing sensitive information, they differ fundamentally in their mechanisms: differential privacy enforces privacy through randomized perturbations, whereas opacity achieves confidentiality by shaping system dynamics and observation structures through policy design.

The proposed information-theoretic opacity is inspired by leakage-minimal channel design in communication \cite{khouzani2017leakage} and will be elaborated in Section~\ref{subsec:problem_statement}. A closely related work is \cite{Parv2013Privacy}, in which the author formulated privacy-preserving control as an optimization problem to  balance utility and information leakage quantified by conditional entropy (equivocation) in \ac{mdp}. Unlike our work, which develops gradient-based optimization for maximizing information-theoretic opacity in stochastic systems, their formulation focuses on static privacy–utility tradeoffs and convex Bellman-equation structures under specific observability assumptions.

Using entropy measure for trajectory obfuscation has also been proposed and studied in \cite{savas2020temporal, molloy2023smoother}. In \cite{savas2020temporal}, the authors add a linear temporal logic constraint to the entropy maximization problem for MDPs. The work in \cite{molloy2023smoother} demonstrated that the conditional entropy of a state trajectory, given both measurements and control inputs, can be decomposed as a cumulative sum of entropy terms. Consequently, standard \ac{pomdp} solvers can be adapted to design controllers that either obfuscate the trajectory to protect sensitive information or enhance its transparency to improve estimation.
 However, in our setting, the observer has no access to the control input. As a result, the conditional entropy of the secret given observation cannot be formulated as a cumulative sum of entropy terms as in \cite{molloy2023smoother}, which prevents  \ac{pomdp} solvers from being applied directly.

\section{Preliminary and Problem Formulation}
\label{sec: preliminary}

\subsection{Preliminaries}
\noindent \textbf{Notation:} The set of real numbers is denoted by $\reals$. Random variables will be denoted by capital letters, and their realizations by lowercase letters (\eg, $X$ and $x$). The probability mass function (pmf) of a discrete random variable $X$ will be written as $p(x)$, the joint pmf of $X$ and $Y$ as $p(x, y)$, and the conditional pmf of $X$ given $Y = y$ as $p(x|y)$ or $p(x|Y = y)$. The sequence of random variables and their realizations with length $T$ are denoted as $X_{[0:T]}$ and $x_{[0:T]}$, respectively. Given a finite and discrete set $\mathcal{S}$, let $\dist(\mathcal{S})$ be the set of all probability distributions over $\mathcal{S}$. The set $\mathcal{S}^{T}$ denotes the set of sequences with length $T$ composed of elements from $\mathcal{S}$, and $\mathcal{S}^\ast$ denotes the set of all finite sequences generated from $\mathcal{S}$.


\paragraph*{The Planning Problem} Consider a stochastic system  modeled as an \ac{mdp} $M=\langle \mathcal{S}, \mathcal{A}, P,\mu_0, R, \gamma \rangle$ where $\mathcal{S}$ is a finite set of states, $\mathcal{A}$ is a finite set of actions, $P:\mathcal{S} \times \mathcal{A}\rightarrow \dist(\mathcal{S})$ is a probabilistic transition function and  $P(s'|s, a)$ is the probability of reaching state $s'$ given that action $a$ is taken at the state $s$,   $\mu_0$ is the initial state distribution, $R: \mathcal{S}\times \mathcal{A}\rightarrow \reals$ is a reward function, and  $\gamma \in [0,1]$ is the discount factor.

 A  Markov policy for $M$ is a function $\pi: S \rightarrow \dist(A)$. A finite-memory policy for $M$ is a function $\pi: (S\times A)^\ast \times S \rightarrow \dist(A)$. 
For a   (Markov/finite-memory) policy $\pi $,  the value function $V  ^{\pi}: \mathcal{S} \rightarrow \reals$ is defined as
\[
V ^{\pi}(s) = \expect_{\pi}[\sum\limits_{t = 0}^{\infty}\gamma^{t}R(S_t, A_t)|S_0 =s],
\] where $\expect_{\pi}$ is the expectation w.r.t. the probability distribution induced by the policy $\pi$ from $M$. We denote the stochastic process induced by the policy $\pi$ as $M_\pi$, and $S_t, A_t$ the $t$-th state and action in the stochastic process $M_\pi$, respectively.

However, unlike classical MDPs, the planning objective is not only to ensure certain reward return or to maximize the total return, but also to protect the confidential information in the system against an eavesdropping observer.  We refer to the planning agent as player 1, or P1, and the observer as player 2, or P2.  We assume that P1 has access to the states. And we assume that the observer P2 has the complete knowledge of MDP, and thereby knows the initial state distribution $\mu_0$. However, the actions of P1 are  non-observable for P2.
P2's observation function is common knowledge, defined as follows:

\begin{definition}[Observation function of P2]
Let $\mathcal{O}$ be a finite set of observations. The state-observation function $\obs: \mathcal{S}  \rightarrow \dist(\mathcal{O})$ of P2 maps a state $s$ to a distribution $\obs(s)$ in the observation space. The action is non-observable.
\end{definition}

Note that the assumption on non-observable actions can be easily relaxed by augmenting the state space with $\mathcal{S}\cup (\mathcal{S}\times \mathcal{A})$ and defining the state-observation function over the augmented state space. This assumption is made only for clarity.

For an MDP $M$ and P2's observation function $\obs$, a policy $\pi$ induces a discrete stochastic process $\{S_t, A_t, O_t, t\in \nat\}$ where each $S_t$ is the random variable of state at time $t$ and $O_t$ is the random variable of observation at time $t$,
\[
P(S_{t+1}=s' | S_{[0:t]} = s,A_{[0:t]} = a_{[0:t]}) = P(s'|s,a_t),
\]
For a Markov policy, 
\[
P(A_t= a | S_{[0:t]} = s_{[0:t]}) = P(A_t= a | S_{ t} = s_{t})=   \pi(a|s_t)\]
For a finite memory policy, 
\begin{equation}
\begin{aligned}
& P(A_t= a | S_{[0:t]} = s_{[0:t]},A_{[0:t-1]}= a_{[0:t-1]}) \\
&= \pi(a|s_{[0:t]},a_{[0:t-1]})
\end{aligned}
\end{equation}
and $P(O_t = o | S_t=s) =  \obs(o|s). $

Next, we will introduce some basic definitions in information theory that are used to quantify the opacity. 

The \emph{entropy of a random variable} $X$ with a countable support $\cal X$ and a probability mass function $p$ is 
\[
H(X) = -\sum_{x\in \mathcal{X}}p(x)\log p(x).
\]
The joint entropy of two random variables $X_1,X_2$ with the same support $\mathcal{X}$ is 
\[
H(X_1,X_2) =  -\sum_{x_1\in \mathcal{X}}\sum_{x_2\in \mathcal{X}} p(x_1,x_2)\log p(x_1,x_2).
\]
The conditional entropy measures the uncertainty about $X_2$ given knowledge of $X_1$.  It is defined as
\[
H(X_2|X_1) =   -\sum_{x_1\in \mathcal{X}}\sum_{x_2\in \mathcal{X}} p(x_1,x_2) \log p(x_2|x_1).
\]
The conditional entropy is also related to the entropy of $X_1$ and the joint entropy.
$$
H(X_2|X_1) = H(X_1,X_2) - H(X_1).
$$
A higher conditional entropy makes it more challenging to infer $X_2$ from $X_1$.


\subsection{Problem Statement}

\label{subsec:problem_statement}

In this section, we develop a mathematical formulation of the opacity-enforcing planning problem, based on the principle of minimal information leakage channel design \cite{khouzani2017leakage} in the context of secret communications. 

In leakage-minimal channel design, the channel is represented by a conditional probability distribution $P(Y|X)$, where $X$ is the input, and $Y$ is the output of the channel. The leakage-minimal design aims to maximize the conditional entropy $H(X|Y)$ so that the information leaked through $Y$ about $X$ is minimized. 

The formulation used in leakage-minimal channel design inspires us to employ the same measure to enforce opacity in \ac{mdp}, where the ``channel'' represents the closed-loop system, i.e., the stochastic process induced by the policy $\pi$. Opacity is measured by the conditional entropy of the automaton or the task state given P2's observations.

\begin{definition}
\label{def:secrets}
Given an \ac{mdp} $M$ and a policy $\pi$, let $M_\pi = \{S_t, A_t, O_t, t\in \nat\}$ be the discrete-time stochastic process induced by $\pi$. Let $T>0$ be finite and $W$ be a set of \emph{secret} states. For any $0\le t\le T$, define
\[ Z_t = \mathbf{1}_W(S_t).\] 
That is, $Z_t$ is the random variable representing if the $t$-th state is in the set of secret states. The random variable $Z_t$ is binary (i.e., $Z_t \in \mathcal{Z} = \{0,1\}$) for each $0 \le t \le T$.
\end{definition}

The conditional entropy of  $Z_t$ given observations $Y= O_{[0:T]}$ can be expressed as 
\begin{multline}
H(Z_t|Y; {M_\pi}) =
-\sum_{z \in \mathcal{Z}}\sum_{y \in \mathcal{O}^T} P^{M_\pi}(Z_t= z, Y= y) \\ \cdot \log P^{M_\pi}(z|Y=y),
\end{multline}
where $P^{M_\pi}(Z_t=z,y)$ is the joint probability of $Z_t=z$ (for $z=0,1$) and observation $y = o_{[0:T]}$ given the stochastic process $M_\pi$,  $P^{M_\pi}(z|Y=y)$ is the conditional probability of $z$ given the observation $y$, and $\mathcal{O}^T$ is the sample space for the observation sequence of length $T$. This conditional entropy \cite{shannonCommunicationTheorySecrecy1949} can be interpreted as the fewest number of subset-membership queries an adversary must make before discovering the secret. We also consider the problem of hiding the initial state. 
Then, nature selects one of the initial states probabilistically according to $\mu_0$, says $s_0$. The agent aims to hide the actual realization of the initial state $s_0$ from the observer.

Based on the notion of information leakage  \cite{yasuokaQuantitativeInformationFlow2010},
the quantitative analysis of confidential information of a dynamical system is defined as the difference between an attacker's \emph{capability in guessing the secret} before and after available observations about the system. Thus, the maximal opacity-enforcement planning with task constraints can be formulated as the following problem:
\begin{problem}[Maximal last-state opacity]
    Given the \ac{mdp} $M$, a set $W$ of secret states, a finite horizon $T$,  compute a policy that maximizes the conditional entropy $H(Z_T|Y; M_\pi)$ between the random variable $Z_T$ and the observation sequence $Y = O_{[0:T]}$ while ensuring the total discounted reward exceeds a given threshold $\zeta$.
\end{problem}
In other words, the objective is to ensure the adversary, with the knowledge of the agent's policy and the observations, is maximally uncertain regarding whether the last state of a finite path is in the secret set. This symmetric notion of opacity implies that opacity will be minimal when the observer is consistently confident that the agent either visited or avoided the secret states, given an observation. This stands in contrast to asymmetric opacity \cite{berardQuantifyingOpacity2015}, which only assesses the uncertainty regarding the visits of secrets. Meanwhile, the agent is required to get an adequate total reward (no less than $\zeta$)  to achieve a satisfactory task performance.

\begin{problem}[Maximal initial-state opacity]
\label{problem:initial-state}
    Given the \ac{mdp} $M$, a finite horizon $T$, and the initial state distribution $\mu_0$.
    Compute a policy that maximizes the conditional entropy $H(S_0  |Y; {M_\pi})$ of the initial state $S_0$ given the observation sequence $Y = O_{[0:T]}$ while ensuring the total discounted reward exceeds a given threshold $\zeta$.
\end{problem}

\section{Synthesizing Constrained Opacity-Enforcement Controllers  }

\subsection{Finite-memory policy}
Beyond Markov policies, we can also employ a \emph{finite-memory} control policy $\pi$. A finite-memory policy is defined as a tuple $\pi \coloneqq \langle \mathfrak{M}, \calS, \calA, \delta_f, \psi, m_0 \rangle$, where:
\begin{inparaenum}
\item $\mathfrak{M}$ is a set of memory states.
\item $\calS$ and $\calA$ denote the sets of inputs and outputs, respectively. Note the state set of the original MDP $M$ is now the input set of the finite-memory policy.
\item $\delta_f: \mathfrak{M}\times \calS \rightarrow \mathfrak{M}$ is a deterministic transition function that updates the memory state based on the current memory state $m$ and input $s\in \calS$, producing the next state $\delta_f(m, s)$.
\item $\psi: \mathfrak{M} \rightarrow \dist(\calA)$ is a probabilistic output function that determines action distributions based on memory states.
\item $m_0$ is the initial memory state.
\end{inparaenum}

Since the Markov policy is a special case of the finite-memory policy, the latter must perform at least as well as the former. To illustrate the advantage of finite-memory policies, we present a simple example demonstrating that a finite-memory policy can achieve superior performance compared to a Markov policy.

\begin{example}
Consider the \ac{mdp} 
$M$ represented as a graph in Fig.~\ref{fig:graph}. The environment includes two sensors: a red sensor at state $q_1$ and a blue sensor at state $q_2$. When the agent is at the state of a sensor
, the observer receives a corresponding observation—``R" for the red sensor and ``B" for the blue sensor—with a probability of $0.5$, while with a probability $0.5$, the observer receives a null observation (``N"). At the initial state, the observer always receives a null observation.
\begin{figure}
\label{fig:graph}
    \centering
    \begin{tikzpicture}[shorten >=1pt,node distance=2.5cm,on grid,auto] 
    \node[state, initial] (0)   {$q_0$}; 
    \node[state,accepting] (1) [above right=of 0, fill=red!60] {$q_1$};   
    \node[state,accepting] (2) [right=of 0, fill=blue!15] {$q_2$}; 
\path[->] 
(0) edge  node[right]{$b:0.5$} (1)
(0) edge  node{$b:0.5$} (2)          

(0) edge [loop above] node{$a:1$} (0) 
(1) edge [loop right]  node{$\truev$} (1)
(2) edge [loop above] node{$\truev$} (2) 
;
\end{tikzpicture}
    \caption{We consider an MDP defined on a graph where the agent has two actions, $a$ and $b$. At the initial state $q_0$, if the agent takes action $a$, it remains in $q_0$. If the agent takes action $b$, it moves with equal probability $0.5$ to either state $q_1$ or state $q_2$. Both $q_1$ and $q_2$ are sink states, meaning that once the agent reaches one of these states, it remains there indefinitely.}
    \label{fig:finite_memory}
\end{figure}
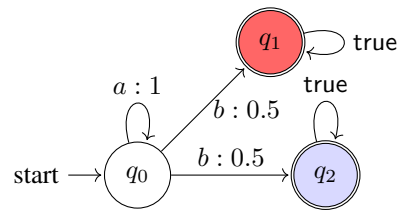
\end{example}

The agent receives a reward of $1$ upon reaching either state $q_1$ or $q_2$. For this example, we ignore the value constraint (i.e., $\zeta = 0$) and set the discount factor to $\gamma = 1$. The trajectory length is restricted to $T = 2$, meaning the agent can take at most two steps. Additionally, the secret state is set to be $q_1$, implying that the agent aims to prevent the observer from determining whether it has visited $q_1$.

\paragraph{Markov Policy}
Under this setting, we only need to consider the policy at the initial state $q_0$ since $q_1$ and $q_2$ are sink states. Let $\pi$ be a Markov policy such that $\pi(a | q_0) = \alpha$ and $\pi(b | q_0) = \beta$. To compute the conditional entropy $H(Z_T | Y; \pi)$, we enumerate all possible state trajectories along with their corresponding probabilities, as shown in Table~\ref{table:trajectory_prob}.

Then, considering the sensor noise, we compute the probability of all possible observations $P(y)$. The probabilities are given by $P(NNN) = \alpha^2 + \frac{1}{2} \alpha \beta + \frac{1}{8} \beta + \frac{1}{8} \beta$, $P(NRN) = \frac{1}{8} - \frac{1}{8} \alpha$, $P(NNR) = \frac{1}{4} \alpha \beta + \frac{1}{8} \beta$, $P(NRR) = \frac{1}{8} \beta$, $P(NBN) = \frac{1}{8} \beta$, $P(NBN) =  \frac{1}{4} \alpha \beta + \frac{1}{8} \beta$, $P(NBB) = \frac{1}{8} \beta$. After obtaining $P(y)$ for all observations $y$, we compute the posterior probability as $P(z_T|y) = P(z_T, y) / P(y)$. The probability $P(1, y)$ is obtained similarly by summing the probabilities of trajectories that reach the secret state $q_1$ and generate the observation $y$. Since $P(0, y) = P(y) - P(1, y)$, we can also obtain the probabilities $P(0, y)$ for all $y$. We do not explicitly list all values of $P(z_T, y)$ here. Note that $\beta = 1 - \alpha$, then the posterior probabilities are
\begin{equation}
P(Z_T = 1|NNN) = \frac{-\frac{1}{4} \alpha^2 + \frac{1}{8} \alpha + \frac{1}{8}}{\frac{1}{2} \alpha^2 + \frac{1}{4} \alpha + \frac{1}{4}},
\end{equation}
$P(Z_T = 1|NRN) = P(Z_T = 1|NRR) =  P(Z_T = 1|NNR) = 1$, $P(Z_T = 1|NNB) =  P(Z_T = 1|NBB) =   P(Z_T = 1|NBN) = 0$. For any $y$, we have $P(Z_T = 1 | y) = 1 - P(Z_T = 0 | y)$. To compute the conditional entropy $H(Z_T | Y; \pi)$, we use the fact that terms where $p \log p = 0$ (when $p = 0$ or $p = 1$) can be omitted. The only contributing terms are for $P(Z_T = 1 | NNN)$ and $P(Z_T = 0 | NNN)$. Thus, the conditional entropy under the Markov policy is given by
\begin{equation}
\label{eq:Markov_entropy}
\begin{aligned}
H(Z_T|Y;\pi) & = - \Big[ (-\frac{1}{4} \alpha^2 + \frac{1}{8} \alpha + \frac{1}{8}) \log(-\frac{1}{4} \alpha^2 + \frac{1}{8} \alpha + \frac{1}{8})\\
&+ (\frac{3}{4} \alpha^2 + \frac{1}{8} \alpha + \frac{1}{8}) \log(\frac{3}{4} \alpha^2 + \frac{1}{8} \alpha + \frac{1}{8}) \\
&- (\frac{1}{2} \alpha^2 + \frac{1}{4} \alpha + \frac{1}{4}) \log(\frac{1}{2} \alpha^2 + \frac{1}{4} \alpha + \frac{1}{4}) \Big].
\end{aligned}
\end{equation}

\begin{figure}[tp!]
\includegraphics[width=\linewidth]{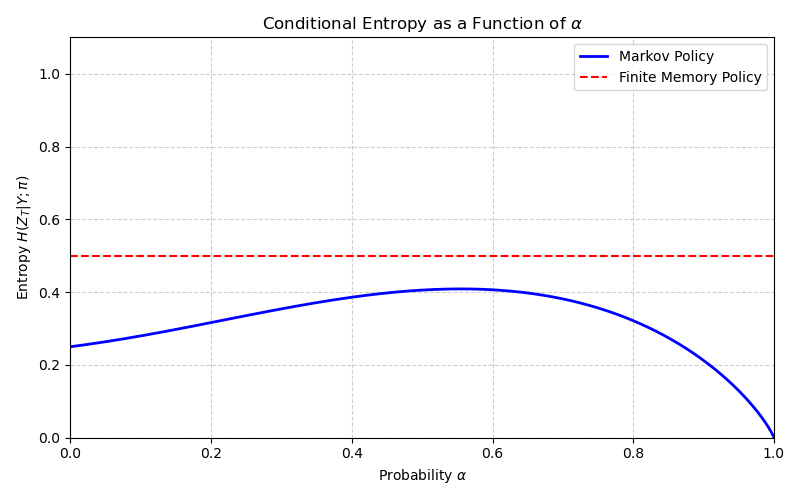}
\centering
\caption{Compare the Markov policy and finite-memory policy.}
\label{fig:comparsion_Markov_finite}
\end{figure}

\begin{table}[t]
\centering
\begin{tabular}{|c|c|c|c|c|c|}
\hline
    &$q_0q_0q_0$ & $q_0q_0q_1$ & $q_0q_0q_2$ & $q_0q_1q_1$ & $q_0q_2q_2$\\ \hline
probability   & $\alpha^2$  & $\frac{1}{2} \alpha \beta$ & $\frac{1}{2} \alpha \beta$ & $\frac{1}{2} \beta$ & $\frac{1}{2} \beta$         \\ \hline
\end{tabular}
\caption{All possible state trajectories and their corresponding probabilities.}
\label{table:trajectory_prob}
\end{table}

\paragraph{Finite-Memory Policy}  We define the finite-memory policy $\pi_f$ that $\pi_f(a|q_0) = 1$, $\pi_f(b|q_0) = 0$, $\pi_f(a|q_0 q_0) = 0$, $\pi_f(b|q_0 q_0) = 1$ and the agent randomly takes actions for other memory states, which is a finite subset of  $Q^\ast$.   Under this policy, the probabilities of all possible state trajectories are given by $P(q_0q_0q_1) = 0.5$ and $P(q_0q_0q_2) = 0.5$. The corresponding observation probabilities, accounting for sensor noise are $P(NNN) = 0.5$, $P(NNB) = 0.25$ , $P(NNR) = 0.25$. The posterior probabilities of the secret state $q_1$ given each observation are $P(Z_T = 1|NNN) = 0.5$, $P(Z_T = 1|NNR) = 1$, $P(Z_T = 1|NNB) = 0$. Thus, the conditional entropy under the finite-memory policy is $H(Z_T|Y;\pi_f) = 0.5$. 

We plot the function from equation~\eqref{eq:Markov_entropy} alongside the value of $H(Z_T | Y; \pi_f)$ in Fig.~\ref{fig:comparsion_Markov_finite}. The figure demonstrates that the finite-memory policy $\pi_f$ outperforms the Markov policy $\pi$, as the conditional entropy under $\pi_f$ is higher, indicating more uncertainty about the secret state after observing the trajectory. 

\subsection{Primal-Dual Policy Gradient for Constrained Minimal Information Leakage}
\label{subsec:def_last_state_opacity}
In this section, we show how to compute a Markov policy for solving the constrained opacity-enforcement problems. The same method can be extended to finite-memory policies by constructing an augmented-state Markov decision process as shown in Appendix~\ref{app:finite_memory_reduce}.

For clarity in notation, we consider a finite state set $\mathcal{S}=\{1,\ldots, N\}$. 
We introduce a set $\{\pi_\theta\mid \theta\in \Theta\}$ of parameterized policies, where $\Theta$ is a finite-dimensional parameter space. 
For any policy $\pi_\theta$ parameterized by $\theta$, the Markov chain induced by $\pi_\theta $ from $M$ is denoted by $M_\theta = \{S_t, O_t, A_t, t\ge 0\}$, where $S_t$ is the random variable for the $t$-th state, and $A_t$ is the random variable for the $t$-th action.

The constrained last-state opacity-enforcement planning problem can be formulated as a constrained optimization problem as follows: 
\begin{equation}
\label{eq:opt_problem_1}
\begin{aligned}
& \optmax_\theta && H (Z_T|Y;\theta) \\
&\optst &&  V(\mu_0,\theta) \ge \zeta,
\end{aligned}
\end{equation}
where $\zeta$ is a lower bound on the value function. The value $V(\mu_0,\theta)$ is obtained by evaluating the policy $\pi_\theta$ given the initial state distribution $\mu_0$, i.e., $V(\mu_0, \theta) \coloneqq V^{\pi_\theta}(\mu_0)$, and $H(Z_T|Y;\theta)  \coloneqq H(Z_T|Y;  M_\theta)$.

\begin{assumption}
\label{assume:local_maxima}
Every local maximum of $V(\mu_0,\theta)$ is a feasible solution. 
\end{assumption}


Assumption~\ref{assume:local_maxima} is a common assumption to ensure the convergence of a gradient algorithm to a feasible solution  \cite{Shie2018rewardconstraint}. And this assumption is guaranteed to hold in some special cases, e.g., the case of softmax policy parameterization as considered later in our paper.


For this constrained optimization problem with an inequality constraint,   we formulate the problem \eqref{eq:opt_problem_1} as the following max-min problem for the associated Lagrangian $L (\theta, \lambda)$:
\begin{equation}
\label{eq:opt_problem_2}
\max_\theta \min_{\lambda \ge 0} L(\theta, \lambda) = H(Z_T|Y;\theta) + \lambda (V(\mu_0, \theta) - \zeta),
\end{equation}
where $\lambda$ is the multiplier. 

In each iteration $k$, the primal-dual gradient descent-ascent algorithm is given as
\begin{equation}
\begin{aligned}
& \theta_{k + 1} = \Gamma_\theta[\theta_k + \eta_k \nabla_\theta L(\theta_k, \lambda_k)], \\
& \lambda_{k + 1} = \Gamma_\lambda[\lambda_k - \kappa_k (V(\mu_0, \theta_k) - \zeta)],
\end{aligned}
\label{eq:primal_dual_algorithm}
\end{equation}
where $\eta_k > 0, \kappa_k > 0$ are step sizes. $\Gamma_\theta$ is a projection operator, which keeps the iterate $\theta_k$ stable by projecting onto a
compact and convex set $[\theta_{\min}, \theta_{\max}]^d$ where $d$ is the dimension. $\Gamma_\lambda$ projects $\lambda_k$ into the range $[0, \lambda_{\max}]$ \footnote{When assumption~\ref{assume:local_maxima} holds, $\lambda_{\max}$ can be set to $\infty$.}. And the gradient of Lagrangian function w.r.t. $\theta$ can be computed as
\begin{equation}
\label{eq:lagrangian}
\nabla_\theta L(\theta, \lambda) = \nabla_\theta H(Z_T|Y;\theta) + \lambda \nabla_\theta V(\mu_0, \theta).
\end{equation}
The term $\nabla_\theta V(\mu_0, \theta)$ can be computed following the classical policy gradient theorem \cite{Sutton1999policy}. In comparison, computing the gradient $\nabla_\theta H(Z_T|Y;\theta)$ of the conditional entropy w.r.t. the policy parameter $\theta$ is nontrivial because the conditional entropy is non-cumulative. The following assumption guarantees the convergence of the algorithm, which will be shown in the section~\ref{sec:convergence_proof}.
\begin{assumption}
\label{assume:step_size}
The step sizes $\eta_k, \kappa_k$ in algorithm~\ref{eq:primal_dual_algorithm} satisfies
\begin{equation}
\begin{aligned}
\sum_{k=0}^\infty \eta_k = \sum_{k = 0}^\infty \kappa_k = \infty, \sum_{k=0}^\infty (\eta_k^2 + \kappa_k^2) < \infty, \frac{\kappa_k}{\eta_k} \to 0.
\end{aligned}
\end{equation}
\end{assumption}
The above assumption is widely used in the stochastic approximation domain~\cite{Borkar2008}. For instance, one can set the $\eta_k = \frac{1}{k^{\xi}}$ where $0.5 < \xi < 1$ and $\kappa_k = \frac{1}{k}$ to satisfy the assumption.

Next, we show how to employ the observable operator methods~\cite{jaeger2000observableoperator} to compute the gradient.

\subsection{Computing the Gradient and Hessian of Conditional Entropy}
\label{subsec:compute_gradient_last}

From the observer's perspective, the stochastic process induced by a Markov policy $\pi_\theta$ is a \ac{hmm}
$HM_\theta = (\mathcal{S}, \mathcal{O}, P_\theta, E)$, where $\mathcal{S}$ is the state space, $\mathcal{O}$ is the observation space, $P_\theta$ is the transition kernel, and $E$ is the emission function, defined by P2's observation function. We first discuss how to compute the gradient and the Hessian of $H(Z_T|Y;\theta)$ with respect to the policy parameter $\theta.$

\subsubsection{The case of last-state opacity}

The conditional entropy of $Z_T$ given an observation sequence $Y$ can be written as
\begin{equation}
\label{eq:HMM_entropy}
H(Z_T|Y;\theta) = - \sum_{y \in \mathcal{O}^T} \sum_{z_T \in \{0,1\}} P_\theta(z_T, y) \log P_\theta(z_T | y), 
\end{equation}
where the probability measure $P_\theta \coloneqq P^{HM_{\theta}}$. The conditional entropy of a binary random variable has a property that $0 \le H(Z_T|Y;\theta) \le 1$. 

\paragraph{Observable Operators}
We propose a novel approach to compute the probability $P_\theta(z_T| y)$ and $P_\theta(y)$ for $z_t\in \{0,1\}$ based on observable operators \cite{jaeger2000observableoperator}. Let the random variable of state, observation, and control action, at time point $t$ be denoted as $S_t, O_t, A_t$, respectively. If we parameterize the policy $\pi$ as $\pi_\theta$, the flipped state transition matrix $\matT^\theta \in \reals^{N \times N}$ would be 
\[
\matT_{i,j}^\theta = P_\theta(S_{t+1} = i|S_t = j) = \sum_{a\in \mathcal{A}} P(i|j,a)\pi_\theta(a|j).
\] 
Let $\matO \in \reals^{M \times N}$ be the observation probability matrix with $\matO_{o,j} =E( o|j)$. 
\begin{definition} 
Given the \ac{hmm} $HM_\theta$, for any observation $o$,
the observable operator $\matA_{o}$ is a matrix of size $N \times N$ with its $ij$-th entry defined as 
$$
\matA_{o}^\theta[i,j] =  \matT_{i, j}^\theta \matO_{o,j} \ ,
$$
which is the probability of transitioning from state $j$ to state $i$ and at the state $j$, an observation $o$ is emitted.
In matrix form, 
\[
\matA_{o}^\theta = \matT^\theta \text{diag}(\matO_{o, 1}, \dots, \matO_{o, N}).
\]
\end{definition}

\begin{proposition}[\cite{jaeger2000observableoperator}, \cite{udupa2025}]
\label{prop:last_opacity_probability} 
The probability of an observation sequence $o_{[0:t]}$ can be written as~
\begin{equation}
\label{eq:matrix_operation}
P_\theta(y) = \mathbf{1}_N^\top \matA_{o_t}^\theta \dots \matA_{o_0}^\theta \mu_0.
\end{equation}
In addition, 
for a fixed state $s_{T + 1} \in \calS$ at time point $T + 1$, we have
\begin{equation}
\label{eq:matrix_operation_sT}
P_\theta(s_{T+1}, y) = \mathbf{1}_{s_{T+1}}^\top \matA_{o_t}^\theta \dots \matA_{o_0}^\theta \mu_0.
\end{equation}
where $\mathbf{1}_{s_{T+1}}$ is a one-hot vector which assigns 1 to the $s_{T+1}$-th entry.
\end{proposition}

Before we introduce the gradient algorithm to enforce the last-state opacity, we first prove that the conditional entropy $H(Z_T|Y;\theta)$ is differentiable under the following assumption. 
\begin{assumption}
\label{assume: differentiable}
The parametrized policy $\pi_\theta(a|s)$ is twice differentiable, i.e., the gradient $\nabla_\theta \pi_\theta(a|s)$ and Hessian $\nabla_\theta^2 \pi_\theta(a|s)$ exists. 
\end{assumption}
The conditional entropy can be written as the expectation of function $-\log P_\theta (z_T|y)$, i.e.,
\begin{equation}
\label{eq:expectation_entropy}
H(Z_T|Y;\theta) = \expect_{(z_T, y)\sim P_\theta(Z_T,Y)} \left[ -\log P_\theta (z_T|y) \right].
\end{equation}
The function $\log P_\theta (z_T|y)$ is continuous and differentiable with respect to $\theta$. 

\begin{proposition}
Under Assumption~\ref{assume: differentiable}, the conditional entropy $H(Z_T|Y;\theta)$ is differentiable with respect to policy parameter $\theta$.
\end{proposition}

\begin{proof}
We first investigate the differentiability of an observable operator $A_{o}^\theta$. The $ij$-th entry of an observable operator $A_{o}^\theta$ is
\begin{equation}
\matA_{o}^\theta[i,j] = \obs(o|s) \sum_{a\in \mathcal{A}} P(s'|s,a) \pi_\theta(a|s).
\end{equation}
The policy $\pi_\theta(a|s)$ is the only term that depends on $\theta$, and it is assumed to be differentiable. The sum and product of differential functions are differentiable. Thus, the observable operator $\matA_{o}^\theta[i,j] $ is differentiable w.r.t $\theta$. 

According to equation~\eqref{eq:matrix_operation_sT}, $P_\theta(s_{T}, y_{-1})$ is differentiable since it is the product of differentiable functions. Then from equation~\eqref{eq:P_zT_yn1}, the probability $P_\theta(Z_{T} = 1| y)$ is differentiable. $P_\theta(Z_{T} = 0| y) = 1 - P_\theta(Z_{T} = 1| y)$ is also differentiable. Finally, by equation~\eqref{eq:HMM_entropy}, the conditional entropy $H(Z_T|Y;\theta)$ is differentiable. 
\end{proof}

To apply the gradient algorithm and analyze its convergence, the gradient and Hessian of the conditional entropy with respect to the policy parameters is required.
In the next proposition, we will give the gradient and Hessian of conditional entropy. 
 The gradient of conditional entropy is calculated as  
\begin{equation}
\label{eq:simplified_gradient}
\begin{aligned}
 &\nabla_\theta H(Z_T|Y;\theta) \\
= & - \sum_{y \in \mathcal{O}^T} \sum_{z_T \in \{0,1\}} \Big[\nabla_\theta P_\theta(z_T, y) \log P_\theta(z_T | y) \\
&+  P_\theta(z_T, y) \nabla_\theta  \log P_\theta(z_T | y)\Big].
\end{aligned}
\end{equation}

Then to calculate the gradient of conditional entropy $\nabla_\theta H(Z_T|Y;\theta)$ by equation~\eqref{eq:simplified_gradient}, we need to calculate the gradient $\nabla_\theta \log P_\theta (z_T|y)$.

By equation~\eqref{eq:matrix_operation}, the gradient $\nabla_\theta P_\theta(y)$ can be calculated as
\begin{equation}
\label{eq:gradient_observations}
\nabla_\theta P_\theta(y) = \sum_{i = 0}^t \mathbf{1}_N^\top \matA_{o_t}^\theta \dots \nabla_\theta \matA_{o_i}^\theta \dots \matA_{o_0}^\theta \mu_0.
\end{equation}
The log gradient $\nabla_\theta \log P_\theta(y)$ is
\begin{equation}
\nabla_\theta \log P_\theta(y) = \frac{1}{\ln 2 \cdot P_\theta(y)} \nabla_\theta P_\theta(y)
\end{equation}

Let $y_{-1} \coloneqq o_{[0:T-1]}$. To calculate the gradient $\nabla_\theta P_\theta(z_T| y)$, we first calculate the probability
\begin{equation}
P_\theta(Z_{T} = 1, y_{-1}) = \sum_{s_{T} \in W} P_\theta(s_{T}, y_{-1}),
\end{equation}
where $W$ is the set of secret defined at the beginning of Defintion~\ref{def:secrets}. Then the conditional probability
\begin{equation}
\label{eq:P_zT_yn1}
P_\theta(Z_{T} = 1| y) = \sum_{s_{T} \in W} \frac{\obs(o_T|s_T) P_\theta(s_{T}, y_{-1})}{P_\theta(y)}.
\end{equation}
Since $P_\theta(z_T, y) = P_\theta(z_T| y) P_\theta(y)$, we can calculate $P_\theta(z_T, y)$ by equation~\eqref{eq:P_zT_yn1}. By this result, we also can calculate the gradient $\nabla_\theta P_\theta(Z_{T} = 1| y)$ as
\begin{equation}
\label{eq:HMM_gradient_P_zT_y}
\begin{aligned}
&\nabla_\theta P_\theta(Z_{T}=1 |y) = \sum_{s_{T} \in W} \obs(o_T|s_T) \nabla_\theta \frac{ P_\theta(s_{T}, y_{-1}) }{P_\theta(y)}\\
&= \sum_{s_{T} \in W} \obs(o_T|s_T) \Big[ \frac{\nabla_\theta  P_\theta(s_{T}, y_{-1})}{P_\theta(y)}  - \frac{ P_\theta(s_{T}, y_{-1})}{ P_\theta^2(y)} \nabla_\theta P_\theta(y) \Big].
\end{aligned}
\end{equation}
To calculate the gradient $\nabla_\theta P_\theta(Z_{T} = 0 |y)$, note that $P_\theta(Z_{T} = 0|y) = 1 - P_\theta(Z_{T} = 1|y)$. Then $\nabla_\theta P_\theta(Z_{T} = 0|y) = -\nabla_\theta P_\theta(Z_{T} = 1|y)$. Since both gradients are equal, we use $\nabla_\theta P_\theta(z_T|y)$ to represent them. 

After computing $\nabla_\theta P_\theta(Z_T = 1 |y), \nabla_\theta P_\theta(Z_T = 0 |y)$, and $\nabla_\theta P(y)$, we can obtain the value of $\nabla_\theta H(Z_T|Y;\theta)$ by equation \eqref{eq:simplified_gradient}. 

It is noted that, though $\mathcal{O}^T$ is a finite set of observations, it is combinatorial and may be too large to enumerate. To mitigate this issue, we can employ sample approximations to estimate   $\nabla_\theta H(Z_T|Y; \theta)$. 

The conditional entropy $H(Z_T|Y;\theta)$ can be written as
\begin{equation}
\begin{aligned}
H (Z_T|Y;\theta) &= \expect_y [ H (Z_T|Y = y;\theta)] \\ &=\expect_y [ \sum_{z_T \in \{0,1\}} P_\theta(z_T| y) \log P_\theta(z_T | y) ].
\end{aligned}
\end{equation}
Then given $M$ sequences of observations $\{y_1, \dots, y_M\}$, we can approximate $ H(Z_T|Y;\theta)$ as
\begin{equation}
\label{eq:HMM_approx_entropy}
H (Z_T|Y;\theta) \approx - \frac{1}{M} \sum_{k=1}^M \sum_{z_T \in \{0,1\}} P_\theta(z_T| y_k) \log P_\theta(z_T | y_k).
\end{equation}
Similarly, we can approximate $\nabla H(Z_T|Y;\theta)$ by
\begin{equation}
\label{eq:HMM_approx_gradient_entropy}
\begin{aligned}
&\nabla_\theta H (Z_T|Y;\theta) \\
&\approx - \frac{1}{M} \sum_{k=1}^M \sum_{z_T \in \{0,1\}} \big[ \log P_\theta(z_T | y_k) \nabla_\theta P_\theta(z_T| y_k) \\
& + P_\theta(z_T| y_k) \log P_\theta(z_T | y_k) \nabla_\theta \log P_\theta(y_k) + \frac{\nabla_\theta P_\theta(z_T | y_k)}{\ln 2} \big].
\end{aligned}
\end{equation}
See the derivation in Appendix~\ref{app:approximation}.
Also, according to~\cite{Sutton1999policy}, we can estimate the gradient of the value function $\nabla_\theta V(\mu_0, \theta)$ by 
\begin{equation}
\label{eq:estimation_gradient_value_function}
\nabla_\theta V(s_0, \theta) \approx \frac{1}{M} \sum_{k=1}^M G_T \nabla_\theta \log \pi_\theta(a|s)].
\end{equation}
\subsubsection{The case with initial-state opacity}

The initial-state opacity is measured by the conditional entropy $H(S_0|Y;\theta)$ of the initial state $S_0$ given by the observation sequence $Y$.
\begin{equation}
\label{eq:HMM_entropy_2}
H(S_0|Y;\theta) = - \sum_{y \in \mathcal{O}^T} \sum_{s_0 \in \mathcal{S}} P_\theta(s_0, y) \log P_\theta(s_0 | y). 
\end{equation}

The constrained opacity-enforcement planning for initial-state opacity can be formulated similarly as a constrained optimization problem:
\begin{equation}
\label{eq:back_opt_problem_1}
\begin{aligned}
& \optmax_\theta && H (S_0|Y;\theta) \\
&\optst && V(s_0,\theta) \ge \zeta,
\end{aligned} 
\end{equation}
where $s_0\in \mathcal{S}$ is the initial state sampled from $\mu_0$. In this optimization problem, we assume that $s_0$ is known to the planning agent but not the observer.  


Following the similar primal-dual optimization approach as in \eqref{eq:primal_dual_algorithm},  we calculate the gradient of $H(S_0|Y;\theta)$ w.r.t. the policy parameter as follows:
\begin{equation}
\label{eq:HMM_gradient_entropy_inital}
\begin{aligned}
 &\nabla_\theta H(S_0|Y;\theta) 
=  - \sum_{y \in \mathcal{O}^T} \sum_{s_0 \in \mathcal{S}} \Big[\nabla_\theta P_\theta(s_0, y) \log P_\theta(s_0 | y) \\
&+  P_\theta(s_0, y) \nabla_\theta  \log P_\theta(s_0 | y)\Big].
\end{aligned}
\end{equation}

The computations of $P_\theta(y)$ and $\nabla_\theta \log P_\theta(y)$ are the same as those for the last-state opacity (see section~\ref{subsec:compute_gradient_last}). The main difference is that,  for the initial-state opacity, we need to obtain the $P_\theta(s_0|y)$ and its derivative w.r.t. $\theta$.


From Bayes' theorem, 
\begin{equation}
\label{eq:bayes_rule}
P_\theta(s_0|y) = \frac{P_\theta(y|s_0) \mu_0(s_0)}{P_\theta(y)}.
\end{equation}
Note that $\mu_0$ is the prior distribution of the initial state, which is known and does not depend on $\theta$. 
Thus, the gradient of $P_\theta(s_0|y)$ w.r.t. $\theta$ is given by
\begin{equation}
\label{eq:gradient_bayes_rule}
\begin{aligned}
\nabla_\theta P_\theta(s_0|y) =  \frac{\mu_0(s_0)}{P_\theta(y)} \nabla_\theta P_\theta(y|s_0) 
- \frac{\mu_0(s_0) P_\theta(y|s_0)}{P_\theta^2(y)} \nabla_\theta P_\theta(y).
\end{aligned}
\end{equation}
And the corresponding log gradient
\begin{equation}
\nabla \log P_\theta(s_0|y) = \nabla_\theta P_\theta(s_0|y) / P_\theta(s_0|y).
\end{equation}
The calculation of gradient $\nabla_\theta P_\theta(y)$ is shown in equation~\eqref{eq:gradient_observations}. For gradient $\nabla_\theta P_\theta(s_0|y)$, we can also compute it using the following proposition.

\begin{proposition}
\label{prop:initial_opacity_probability} 
For a fixed initial state $s_{0} \in \calS$,
\begin{equation}
\label{eq:matrix_operation_s0}
P_\theta(y|s_0) = \mathbf{1}_{N}^\top \matA_{o_t}^\theta \dots \matA_{o_0}^\theta \mathbf{1}_{s_0}.
\end{equation}
where $\mathbf{1}_{N}$ is a vector of size $N$ with all entries equal to one and $\mathbf{1}_{s_0}$ is a one-hot vector which assigns 1 to the $s_0$-th entry.
\end{proposition}

\begin{proof}
The equation~\eqref{eq:matrix_operation_s0}
is derived by replacing the initial distribution $\mu_0$ with the one-hot distribution in the equation~\eqref{eq:matrix_operation}. 
\end{proof}

Then by equation~\eqref{eq:matrix_operation_s0}, the gradient $\nabla_\theta P_\theta(y|s_0)$ can be calculated as
\begin{equation}
\nabla_\theta P_\theta(y|s_0) = \sum_{i = 0}^t \mathbf{1}_N^\top \matA_{o_t}^\theta \dots \nabla_\theta \matA_{o_i}^\theta \dots \matA_{o_0}^\theta \mathbf{1}_{s_0}.
\end{equation}

  In this way, $H(S_0|Y;\theta)$, $\nabla_\theta H(S_0|Y;\theta)$ and  $ \nabla_\theta^2 H(S_0|Y;\theta) $  can be computed exactly.  Similar sample approximations for the last-state-opacity case can be used to estimate these values and thus we omit the derivation.  


With the above step of calculating the gradient of conditional entropy $H (S_0|Y;\theta)$ w.r.t. $\theta$, we can then employ the primal-dual approach to solve a (locally) optimal solution to Problem~\ref{problem:initial-state}.


\section{Convergence Guarantee of the Primal-Dual Policy Gradient Method}
\label{sec:convergence_proof}
Next, we will show that the conditional entropy $H(Z_T|Y;\theta)$ is Lipschitz continuous and L-smooth in the policy parameter $\theta$, provided that the following condition holds for the policy function space.


\begin{assumption}
\label{assume:bounded-policy-gradient}
 For any $(s,a)\in \mathcal{S} \times \mathcal{A}$, both $\nabla_\theta \log \pi_\theta(a | s) $ and $\nabla_\theta^2 \log \pi_\theta(a| s) $ are bounded. 
\end{assumption}

The above assumption is mild, as it is satisfied by several common policy parameterizations. For example, the gradient and Hessian of the softmax policy are bounded when the parameter $\theta$ lies in a compact set~\cite{wei2025activeinferenceincentivedesign}, as specified in equation~\eqref{eq:primal_dual_algorithm}.

\begin{lemma}
\label{lem:first_order}
Under Assumption~\ref{assume:bounded-policy-gradient}, the gradients $\nabla_\theta P_\theta(s_T, y)$ and $\nabla_\theta P_\theta(y)$    are bounded.
\end{lemma}
\begin{proof}
First, we introduce the concept of forward $\alpha$ messages from \ac{hmm}s \cite{baum1970maximization}. For clarity in notation, we consider a finite state set $\mathcal{S}=\{1,\ldots, N\}$. The emission probability distribution $b_{i}(o)= \obs(o|i)$. Given a fixed observation sequence $y$, the forward $\alpha$ message at the time step $t$ for a given state $j \in \calS$ is, 
\begin{equation}
\label{eq:forward_path_probability}
\alpha_t(j, \theta) \coloneqq P_\theta(y, S_t = j),
\end{equation}
which is the joint probability of receiving observation $y$ and arriving at state $j$ at the $t$-th time step. Note that $\nabla_\theta P_\theta(s_T, y) = \nabla_\theta \alpha_T(j, \theta)$. Thus, we want to prove $\nabla_\theta \alpha_T(j, \theta)$ is bounded. 

The forward $\alpha$ messages can be calculated recursively by the following equations \cite{baum1970maximization}: for $t=1,\ldots T$, 
\begin{equation}
\label{eq:update_forward_path_probability}
\alpha_t(j, \theta) = \sum_{i \in \mathcal{S}} \alpha_{t-1}(i, \theta) P_\theta(i,j) b_j(o_t), 
\end{equation}
 
The initial forward $\alpha$ message is defined jointly by the initial observation and the initial state distribution,
$\alpha_0(j,\theta) = \mu_0(j) b_j(o_0), 0 \le j \le N$. 

Then we can also compute these gradients using the following recursive computation based on  \eqref{eq:update_forward_path_probability}: For $1 \le t \le T$,
\begin{equation}
\label{eq:update_gradient_forward_path_probability}
\begin{aligned}
&\nabla_\theta \alpha_t(j, \theta) = \sum_{i=1}^N P_\theta(i,j) b_j(o_t)\nabla_\theta \alpha_{t-1}(i, \theta) \\
&+ \sum_{i=1}^N \alpha_{t-1}(i, \theta) b_j(o_t) \nabla_\theta P_\theta(i,j),
\end{aligned}
\end{equation}
and $\nabla_\theta \alpha_0(j, \theta) = 0$ since $\alpha_0(j,\theta)$ does not depend on $\theta$. In addition, the gradient
\begin{equation}
\label{eq:grad_P_i_j_calculation}
\begin{aligned}
\nabla_\theta P_\theta(i,j) &= \sum_{a \in \mathcal{A}} P(j|i,a) \nabla_\theta \pi_\theta(a|i)\\
&= \sum_{a \in \mathcal{A}} P(j|i,a) \pi_\theta(a|i) \nabla_\theta \log \pi_\theta(a|i)
\end{aligned}
\end{equation}
can be computed using the current policy $\pi_\theta$ and the gradient of the policy w.r.t. $\theta$. Since $\nabla_\theta \log \pi_\theta(a|i)$ is bounded, $\nabla_\theta P_\theta(i,j)$ is bounded by the triangle inequality, i.e.,
\begin{equation}
\| \nabla_\theta P_\theta(i,j) \| \le \sum_{a \in \mathcal{A}} P(j|i,a) \pi_\theta(a|i) \| \nabla_\theta \log \pi_\theta(a|i) \|.
\end{equation}

And according to the recursion~\eqref{eq:update_gradient_forward_path_probability}, $\nabla_\theta \alpha_T(j, \theta)$ is bounded because it is a finite summation of bounded gradients. Therefore, $\nabla_\theta P_\theta(s_T, y)$ is bounded.

Since the probability of receiving observation $y$ is $P_\theta(y) = \sum_{i \in\mathcal{S}} \alpha_T(i, \theta)$ and  $\nabla_\theta \alpha_T(i, \theta)$ is bounded, the gradient  $\nabla_\theta P_\theta(y) = \sum_{i \in \mathcal{S}} \nabla_\theta \alpha_T(i, \theta)$ is bounded.
\end{proof}

Next, we need to calculate some Hessian matrices for the proof of Lipschitz continuity. We list them here. Given an observation sequence $y$ and $z_T\in W$ such that  $P_\theta(z_T|y) \neq 0$ and $P_\theta(y) \neq 0$. We can derive the Hessian of conditional entropy from equation~\eqref{eq:simplified_gradient}.
\begin{equation}
\begin{aligned}
\label{eq:simplified_hessian}
 &\nabla_\theta^2 H(Z_T|Y;\theta) \\
= & - \sum_{y \in \mathcal{O}^T} \sum_{z_T \in \{0,1\}} \Big[\nabla_\theta^2 P_\theta(z_T, y) \log P_\theta(z_T | y) \\
&+ \nabla_\theta P_\theta(z_T, y) \nabla_\theta  \log P_\theta(z_T | y)^\top 
\\
&+ P_\theta(z_T, y) \nabla_\theta^2 \log P_\theta(z_T | y)
\\
&+ \nabla_\theta  \log P_\theta(z_T | y) \nabla_\theta P_\theta(z_T, y)^\top \Big].
\end{aligned}
\end{equation}
And the Hessian of $\log P_\theta(z_T|y)$ can be calculated as
\begin{equation}
\label{eq:hessian_log_P_zT_y}
\begin{aligned}
\nabla_\theta^2 \log P_\theta(z_T|y) = \frac{\nabla_\theta^2 P_\theta(z_T|y)}{P_\theta(z_T|y)}  - \frac{\nabla_\theta P_\theta(z_T|y) \nabla_\theta P_\theta(z_T|y)^\top}{P_\theta^2(z_T|y)} 
\end{aligned}
\end{equation}
where
\begin{equation}
\label{eq:hessian_p_z_y_prop}
\begin{aligned}
\nabla_\theta^2 & P_\theta(z_T|y) = \sum_{s_{T} \in W} \Big[ \frac{\nabla_\theta^2 P_\theta(s_{T}, y)}{P_\theta(y)} - \frac{\nabla_\theta P_\theta(s_{T}, y) \nabla_\theta P_\theta(y)^\top}{P_\theta^2(y)} \\
& - \frac{\nabla_\theta P_\theta(y) \nabla_\theta P_\theta(s_{T}, y)^\top}{P_\theta^2(y)} - \frac{P_\theta(s_{T}, y) \nabla_\theta^2 P_\theta(y)}{P_\theta^2(y)} \\ 
& + \frac{2 P_\theta(s_{T}, y) \nabla_\theta P_\theta(y) \nabla_\theta P_\theta(y)^\top}{P_\theta^3(y)} \Big].
\end{aligned}
\end{equation}

\begin{lemma}
\label{lem:second_order}
Under Assumption~\ref{assume:bounded-policy-gradient}, the Hessian of probability $\nabla_\theta^2 P_\theta(s_T, y)$ and $\nabla_\theta^2 P_\theta(y)$ are bounded.
\end{lemma}
\begin{proof}
Follow the proof of Lemma~\ref{lem:first_order}, the Hessian matrix $\nabla_\theta^2 P_\theta(s_T, y) = \nabla_\theta^2 \alpha_T(j, \theta)$. Calculate the gradient of equation~\eqref{eq:update_gradient_forward_path_probability}, we obtain
\begin{equation}
\label{eq:update_hessian_forward_path_probability}
\begin{aligned}
&\nabla_\theta^2 \alpha_t(j, \theta) = \sum_{i=1}^N  b_j(o_t)\nabla_\theta \alpha_{t-1}(i, \theta) \nabla_\theta P_\theta(i,j)^\top \\
& + \sum_{i=1}^N P_\theta(i,j) b_j(o_t)\nabla_\theta^2 \alpha_{t-1}(i, \theta) \\
& + \sum_{i=1}^N \alpha_{t-1}(i, \theta) b_j(o_t) \nabla_\theta^2 P_\theta(i,j),
\end{aligned}
\end{equation}
and $\nabla_\theta^2 \alpha_0(j, \theta) = 0$ since $\nabla_\theta \alpha_0(j,\theta) = 0$. In addition, the Hessian
\begin{equation}
\label{eq:hessian_P_i_j_calculation}
\begin{aligned}
\nabla_\theta^2 & P_\theta(i,j) = \sum_{a \in \mathcal{A}} P(j|i,a) \nabla_\theta^2 \pi_\theta(a|i)\\
&= \sum_{a \in \mathcal{A}} P(j|i,a) \nabla_\theta [\pi_\theta(a|i) \nabla_\theta \log \pi_\theta(a|i)]\\
&= \sum_{a \in \mathcal{A}} P(j|i,a) \pi_\theta(a|i) \nabla_\theta \log \pi_\theta(a|i) \nabla_\theta \log \pi_\theta(a|i)^\top \\
& + \sum_{a \in \mathcal{A}} P(j|i,a) \pi_\theta(a|i) \nabla_\theta^2 \log \pi_\theta(a|i).
\end{aligned}
\end{equation}

Since $ \nabla_\theta \log \pi_\theta(a | s) $ and $ \nabla_\theta^2 \log \pi_\theta(a | s) $ are bounded, then $\nabla_\theta P_\theta(i,j)$ is bounded. According to the recursion~\eqref{eq:update_hessian_forward_path_probability},
$\nabla_\theta^2 \alpha_T(j, \theta)$ is finite because it is a finite summation of bounded gradients and Hessian matrices. Therefore, $\nabla_\theta^2 P_\theta(s_T, y)$ is bounded.
Further, since $\nabla_\theta^2 \alpha_T(k, \theta)$ is bounded for each $k \in \mathcal{S}$, the Hessian $\nabla_\theta^2 P_\theta(y) = \sum_{k \in \mathcal{S}} \nabla_\theta^2 \alpha_T(k, \theta)$ is bounded.
\end{proof}

\begin{theorem}
\label{thm:last_state_convergence}
Under Assumption~\ref{assume:bounded-policy-gradient}, 
the entropy $H(Z_T|Y; \theta)$ is Lipschitz-continuous and Lipschitz-smooth in $\theta$.
\end{theorem}
\begin{proof}
We prove that the entropy $H(Z_T|Y; \theta)$ is Lipschitz-continuous by proving the gradient $\nabla_\theta H(Z_T|Y; \theta)$ is bounded. From Lemma~\ref{lem:first_order}, $\nabla_\theta P_\theta(s_T, y)$ and $\nabla_\theta P_\theta(y)$ are bounded. 

Then since the gradient of $\log P_\theta(z_T|y)$ can be calculated as
\begin{equation}
\label{eq:grad_log_P_z_y}
\nabla_\theta \log P_\theta(z_T|y) = \nabla_\theta P_\theta(z_T|y) / P_\theta(z_T|y),
\end{equation}
we can derive that $\nabla_\theta \log P_\theta(z_T|y)$ is bounded. According to equation~\eqref{eq:simplified_gradient}, $\nabla_\theta H(Z_T|Y; \theta)$ is clearly bounded. 

Similarly, we can prove that $H(Z_T|Y; \theta)$ is Lipschitz-smooth by proving that the Hessian $\nabla_\theta^2 H(Z_T|Y; \theta)$ is bounded.  From Lemma~\ref{lem:second_order}, $\nabla_\theta^2 P_\theta(s_T, y)$ and $\nabla_\theta^2 P_\theta(y)$ are bounded. Thus, given \eqref{eq:hessian_log_P_zT_y} and \eqref{eq:hessian_p_z_y_prop} we can conclude that   $\nabla_\theta^2 \log P_\theta(z_T|y)$ is bounded.
According to equation~\eqref{eq:simplified_hessian}, $\nabla_\theta^2 H(Z_T|Y; \theta)$ is also bounded.  
\end{proof}

Next, we show that similar properties hold for initial-state opacity.

\begin{theorem}
\label{thm:initial_state_convergence}
Under Assumption~\ref{assume:bounded-policy-gradient}, 
the entropy $H(S_0|Y; \theta)$ is Lipschitz-continuous and $L$-smooth in $\theta$.
\end{theorem}

Due to the similarity with the last-state opacity case, we defer the proof to Appendix~\ref{app:proof-initial-state}.

\begin{theorem}
For the last-state opacity enforcement problem, under assumption~\ref{assume:local_maxima} and \ref{assume:step_size}, the iterates $(\theta_k, \lambda_k)$ of the algorithm in~\eqref{eq:primal_dual_algorithm} with gradient approximations~\eqref{eq:HMM_approx_gradient_entropy}, ~\eqref{eq:estimation_gradient_value_function}
converge to a fixed point (local maximum) almost surely. The fixed point is a feasible solution to problem~\eqref{eq:opt_problem_1}. 
\end{theorem}

\begin{proof}
We follow the proof in \cite{Shie2018rewardconstraint}, which utilizes techniques from stochastic approximation methods \cite{Borkar2008}. We first prove that the algorithm converges to a local saddle point of the Lagrangian in~\eqref{eq:lagrangian}. 

Due to the timescale separation in assumption~\ref{assume:step_size}, the value of $\lambda$ (updated on the slower timescale) can be regarded as a constant in the recursion of $\theta$. The following ODE governs the evolution of $\theta$:
\begin{equation}
\label{eq:ODE_theta}
\Dot{\theta}_t = \Gamma_\theta(\nabla_\theta L(\theta_t, \lambda))
\end{equation}
where $ \Gamma_\theta$ projects $\theta$ to a compact and convex set $\Theta$. As $\lambda$ is considered constant, the process over $\theta$ is:
\begin{equation}
\theta_{k + 1} = \Gamma_\theta[\theta_k + \eta_k \nabla_\theta L(\theta_k, \lambda)].
\end{equation}
Thus, the recursion of $\theta$ in \eqref{eq:primal_dual_algorithm} can be seen as a discretization of the ODE~\eqref{eq:ODE_theta}. By theorem~\ref{thm:last_state_convergence}, $\nabla_\theta H(Z_T|Y;\theta)$ is Lipschitz-continous (w.r.t. $\theta$). By policy gradient theorem \cite{Sutton1999policy}, the gradient of the value function can be written as
\begin{equation}
\nabla_\theta V(s_0, \theta) = \expect_\theta[G_T \nabla_\theta \log \pi_\theta(a|s)]
\end{equation}
where $G_T = \sum_{t = 0}^T \gamma^t R(s_t, a_t)$ is the total return. 
Under assumption~\ref{assume:bounded-policy-gradient}, it is clear that $\nabla_\theta V(s_0, \theta) $ is Lipschitz-continous (w.r.t. $\theta$). Therefore, $\nabla_\theta L(\theta_t, \lambda)$ is Lipschitz-continuous which satisfies the assumption of functions in chapter 6 of \cite{Borkar2008}. Finally, using the standard stochastic approximation arguments from \cite{Borkar2008} concludes the convergence of $\theta$-recursion. 

The process governing the evolution of $\lambda$:
\begin{equation}
\lambda_{k + 1} =  \Gamma_\lambda[\lambda_k - \kappa_k (V(\mu_0, \theta(\lambda_k)) - \zeta)]
\end{equation}
where $\theta(\lambda_k)$ is the limiting point of the $\theta$-recursion corresponding to $\lambda_k$, can be seen as the following ODE:
\begin{equation}
\label{eq:ODE_lambda}
\Dot{\lambda}_t = \Gamma_\lambda[\zeta - V(\mu_0, \theta(\lambda_t)))].
\end{equation}
As shown in Chapter~6 of \cite{Borkar2008}, $(\lambda_n, \theta_n)$ converges to the internal chain transitive invariant sets of the ODE~\eqref{eq:ODE_lambda}, $\dot{\theta}_t = 0$.
Thus, $(\lambda_k, \theta_k) \to \{ (\lambda(\theta), \theta) : \theta \in \mathbb{R}^D \}$ almost surely where $D$ is the dimension of $\theta$.

Finally, as seen in Theorem~2 of Chapter~2 of \cite{Borkar2008}, $\theta_k \to \theta^\star$ where $\nabla_\theta L(\theta^\star, \lambda) = 0$ almost surely, then $\lambda_k \to \lambda(\theta^\star)$ almost surely, which completes the proof of convergence to the local saddle point. 

Assumption~\ref{assume:local_maxima} states that any local maximum $\theta$ satisfies the constraints, i.e., $V(s_0, \theta) \ge \zeta$; additionally, It has been shown that first-order methods (stochastic gradient methods) converges almost surely to a local maxima (avoiding saddle points) \cite{Lee2019avoid}. Hence, for an unbounded Lagrange multiplier ($\lambda^\star$ can be infinite) 
, the process converges to a fixed point ($(\theta^\star(\lambda^\star), \lambda^\star)$) which is a feasible solution. 



\end{proof}

\begin{remark}
For the initial-state opacity enforcement problem, a similar convergence proof can be derived given the property in Theorem~\ref{thm:initial_state_convergence}. 
\end{remark}

\section{Synthesizing Maximally Opacity-Enforcement Controllers For Language-Based Opacity}
In this section, we will introduce information-theoretic language-based opacity, a type of opacity in which a language or a temporal logic formula defines the secret. We will show that language-based opacity can be reduced to last-state opacity in the sense of information theory based on the relation between temporal logic and automaton, where the secret is defined as the automaton state that indicates the satisfaction of a linear temporal logic formula.

\subsection{Labeled Markov Decision Process and Linear Temporal Logic Formula }

We consider a probabilistic planning problem modeled as a Markov decision process, augmented with a set of atomic propositions and a labeling function, which are used to define the secret of the system.

A \ac{lmdp} is an MDP $M=\langle \mathcal{S}, \mathcal{A}, P,\mu_0, R, \gamma, \calAP, L \rangle$  with  two additional components:
\begin{itemize} \item $\calAP$ is the set of atomic propositions; and
    \item $L: S\rightarrow 2^{\calAP}$ is the labeling function that maps a state to a set of atomic propositions that evaluate true at that state. \end{itemize}

For a finite play $\rho = s_0 a_0 s_1 a_1 s_2\ldots s_n $, 
the labeling of the play $\rho$, denoted $L(\rho)$, is defined as $L(\rho) = L(s_0)L(s_1)\ldots L(s_n)$. That is, the labeling function omits the actions from the play and applies to states only.

\paragraph*{P1's secret temporal objective}
P1 has a secret modeled by a \ac{ltl} formula and aims to obfuscate the information so that P2 is uncertain about P1's progress with respect to satisfying the temporal objective.

%

We omit the syntax and semantics of \ac{ltlf}, which can be found in \cite{de2013linear}. It is known that 
the language of \ac{ltlf} formula $\varphi$ can be represented by the set of words accepted by an automaton $\mathcal{A} = (\mathcal{Q}, \Sigma, \delta, \init, F )$ in which (1) $\mathcal{Q}$ is the set of states; (2) $\Sigma$ is the alphabet (set of input symbols); (3) $\delta: \mathcal{Q} \times \Sigma \rightarrow \mathcal{Q}$ is a deterministic transition function and is complete \footnote{For any $\mathcal{Q} \times \Sigma$, $\delta(q,\sigma)$ is defined. An incomplete transition function can be completed by adding a sink state and redirecting all undefined transitions to that sink state.}; (4) $\init$ is the initial state; and (5) $F \subseteq \mathcal{Q}$ is the set of accepting states. 

The transition function $\delta$ is extended as $\delta(q, \sigma \cdot w) = \delta(\delta(q,\sigma), w)$ where the state $q \in \mathcal{Q}$ and input $\sigma \in \Sigma$. %
A word $w = w_0 w_1 \ldots w_n \in \Sigma^\ast$ is accepted by $\calA$ if and only if $\delta(\init, w)\in F$. 
The set of words accepted by $\calA$ is called the language of $\calA$, denoted by $\mathcal{L}(\calA)$.
Formally, $\mathcal{L}(\calA) = \{ w \in \Sigma^* \mid \delta(q_0, w) \in F \}$.
For notation simplicity, let $\Sigma \coloneqq 2^\calAP$. We assume the DFA for the \ac{ltlf} formula $\varphi$  is known to the observer. For example, P1 and P2 can construct the minimal \ac{dfa} for the given formula.

An informal problem statement is as follows.
\begin{problem}
Given an LMDP $M$ and a secret LTL formula $\varphi$,   design a policy on the LMDP to conceal the information from P2 regarding P1's progress in satisfying the secret $\varphi$, subject to a constraint on the expected total reward given a finite horizon. 
\end{problem}

\subsection{Reduction: From Language-Based Opacity to Last-State Opacity}

\begin{definition}[Product MDP]
\label{def:product_mdp}
Given the LMDP $M=\langle \mathcal{S}, \mathcal{A}, P,\chi_0, R, \gamma, \calAP, L \rangle$, and a \ac{dfa} $\mathcal{A} = (\mathcal{Q},2^\calAP, \delta, \init, F)$ describing P1's secret objective $\varphi$, the product MDP a tuple 
\[
\mathcal{M} = (\mathcal{V}, \mathcal{A}, \Delta, v_0, \mathcal{R})
\]
in which
\begin{itemize}
    \item $\mathcal{V} = \mathcal{S} \times \mathcal{Q}$ is the state space, where each state $(s,q)$ includes an MDP state $s \in \mathcal{S}$, and an automaton state $q \in \mathcal{Q}$;
    \item $\mathcal{A}$ is the P1's action space, same as in $M$;
    \item $\Delta: \mathcal{V} \times \mathcal{A} \rightarrow \dist(\mathcal{V})$ is the probabilistic transition function such  given a state $v \coloneqq (s,q) \in \mathcal{V}$, for any state $s' \in \mathcal{S}$ and an action $a \in \mathcal{A}$, 
    \[
    \Delta((s',q')| (s, q), a) =  P(s'|s,a)
    \]
    where $q'=\delta(q, L(s'))$. Else, $\Delta((s',q')| (s, q), a) = 0$.
    \item $\chi_0$ is the initial state distribution that the initial state $v_0 \coloneqq (s_0,q_0) \sim \chi_0$ with $s_0 \sim \mu_0, q_0 = \delta(\init, L(s_0))$. 
    \item $\mathcal{R} \colon \mathcal{V} \times \mathcal{A} \to \reals$ is the reward function that describes the planning objective:
    \begin{equation}
    \mathcal{R}((s,q), a) = R(s, a). 
    \end{equation}
\end{itemize}
\end{definition} 


Given the product MDP $\mathcal{M}$,  a policy $\pi$ induces a discrete stochastic process $\mathcal{M}_\pi = \{V_t \coloneqq (S_t, Q_t), O_t, t\in \nat\}$. In the case of a Markov policy, 
\[
P(V_{t+1}=v' | V_t = v) = \sum_{a\in \mathcal{A}} \Delta(v'|v,a)\pi(a|v),
\]
and $P(O_t = o | V_t=v) =  \obs(o| s)$, where $v = (s,q)$. 

The conditional entropy of the automaton state $Q_t$ for $0 \le t \le T$ given a finite observation $Y= O_{[0:T]}$ is defined as  
\begin{equation}
\begin{aligned}
H(Q_t|Y; {\mathcal{M}_\pi}) =
-\sum_{q \in \mathcal{Q}}\sum_{y \in \mathcal{O}^T} P^{\mathcal{M}_\pi}(Q_t= q, Y= y) \\ \cdot \log P^{\mathcal{M}_\pi}(Q_t=q|Y=y),
\end{aligned}
\end{equation}
where $P^{\mathcal{M}_\pi}(Q_t=q,y)$ is the joint probability of $Q_t=q$ (for $q \in \mathcal{Q}$) and observation $y=o_{[0:t]}$ given the stochastic process $\mathcal{M}_\pi$ and  $P^{\mathcal{M}_\pi}(Q_t= q|Y=y)$ is the conditional probability of $q$ given the observation $y$ and $\mathcal{O}^T$ is the sample space for the observation sequence of length $T$.  

In addition, P1's value function $V^{\pi}: \mathcal{V} \rightarrow \reals$ is defined as
\[
V ^{\pi}(v) = \expect_{\pi}[\sum\limits_{k = 0}^{\infty}\gamma^{k}\mathcal{R}(V_k, \pi(V_k))|V_0 =v],
\]
where $\expect_{\pi}$ is the expectation w.r.t. the probability distribution induced by the Markov policy $\pi$ from $\mathcal{M}$.


By defining the conditional entropy of the automata state, the maximal language-based opacity-enforcement planning problem can be formulated as a constrained optimization problem as follows:
\begin{equation}
\label{eq:opt_problem_3}
\begin{aligned}
& \max_\theta && H (Q_T|Y;\theta) \\
&\optst: &&  V(\mu_0,\theta) \ge \zeta 
\end{aligned}
\end{equation}
where $\zeta$ is a lower bound on the value function. The value $V(\mu_0,\theta)$ is obtained by evaluating the policy $\pi_\theta$ given the initial state distribution $\mu_0$, i.e., $V(\mu_0, \theta) \coloneqq V^{\pi_\theta}(\mu_0)$, and $H(Q_T|Y;\theta) \coloneqq H(Q_T|Y;  \mathcal{M}_\theta)$. This optimization problem has a similar form as the maximizing last-state opacity in section~\ref{subsec:def_last_state_opacity}. We can use the same primal-dual policy gradient algorithm to solve it.

\section{Experiment Evaluation}
In this section, we demonstrate the effectiveness of our algorithm through several examples. Specifically, we evaluate the last-state and initial-state opacity maximization algorithms on a grid-world example, and the language-based opacity algorithm on a graph-based example.
\subsection{Grid World Example for State-Based Opacity}
The effectiveness of the proposed optimal state-based opacity-enforcement planning algorithms 
is illustrated through a stochastic grid world example shown in Fig.~\ref{fig:grid_world_env}. In this example, we focus on optimizing the last-state opacity and initial-state opacity, respectively. The details of the environment setting are outlined in Fig.~\ref{fig:grid_world_env}. For perception, five sensors are placed on the grid with distinct ranges indicated by the blue, red, yellow, green, and grey areas in the picture. 

As P1 enters the range of a sensor, the observer receives corresponding observations (``A", ``B", ``C", ``D",  ``E", respectively) with probability $p=0.7$ and a null observation (``$N_0$") with probability $1-p=0.3$, 
attributed to the false negative rate of the sensors. If P1 is outside the range of all sensors, the observer consistently receives the null observation “$N_0$.”  
\begin{figure}[tp!]
\includegraphics[width=0.8\linewidth]{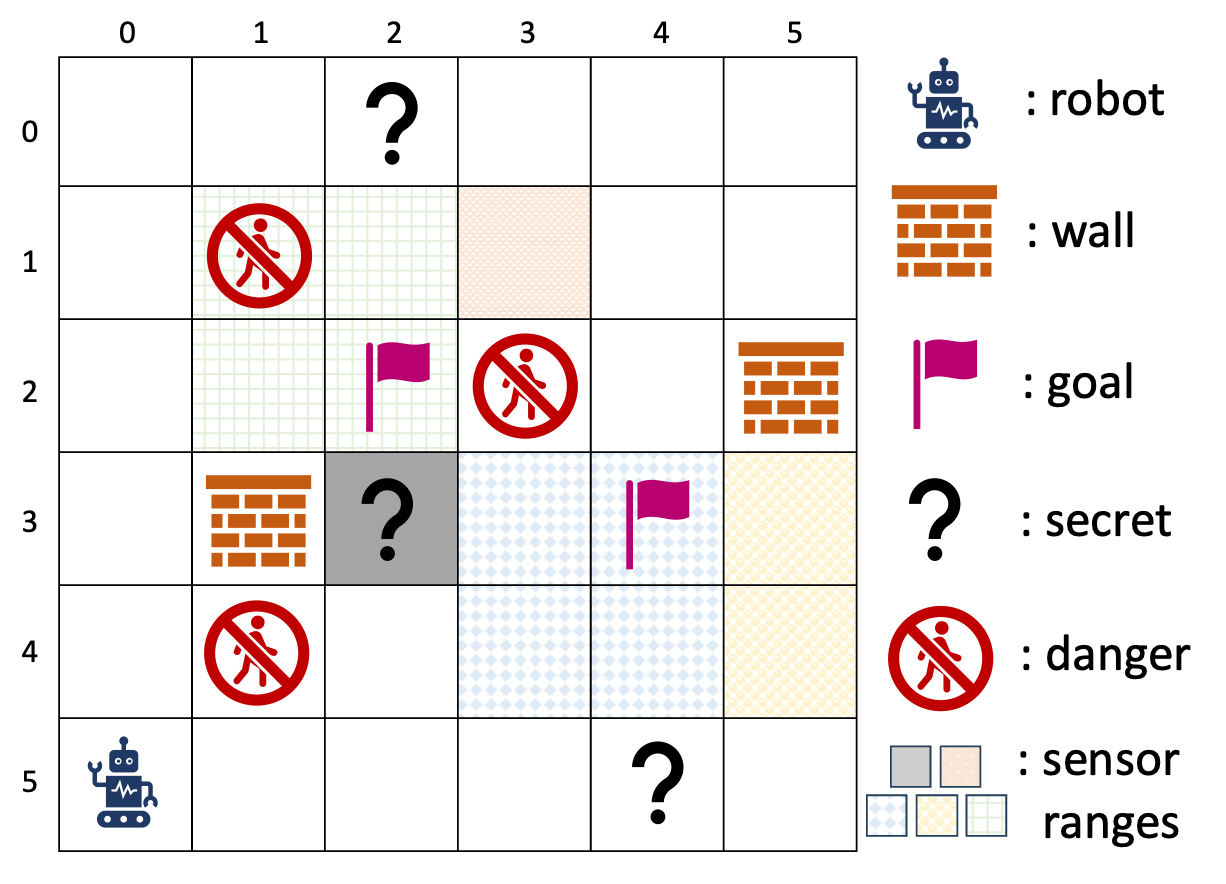}
\centering
\caption{The blue robot is P1 (the agent). P1 can move in four compass directions (north, south, east, west) or remain stationary. However, the dynamics of movement are stochastic. When the robot moves in a specific direction, there is a 0.1 probability that it will also move in the nearest two directions. For instance, if the robot moves east, there is a 0.1 probability of it moving north and a 0.1 probability of it moving south, as illustrated in the image. If the robot hits the boundary or walls, it stays put. And it becomes immobilized if it enters dangerous cells.}
\label{fig:grid_world_env}
\end{figure}

The question marks on the grid represent the secret states for P1, while the flags denote the goal states for P1.  We set the reward of reaching a goal to be $1$. The goal and secret states are not sink states, meaning P1 can revisit these states multiple times. 

We will employ the soft-max policy parameterization, i.e.,
\begin{equation}
    \pi_{\theta}(a|s) = \frac{\exp(\theta_{s,a})}{\sum_{a' \in \mathcal{A}}\exp(\theta_{s, a'})},
\end{equation}
where $\theta\in \reals^{
|S\times A|}$ is the policy parameter vector.
The softmax policy has good analytical properties, including completeness and differentiability. Under the softmax policy parameterization, a feasible policy always exists.

We will use a randomly selected policy (by randomly selecting a policy parameter $\theta$) as the initial policy for future experiments. We restrict the value of $\theta_{s,a}$ within $[-700, 700]$ by projection in~\eqref{eq:primal_dual_algorithm}. Because 1) the evolution of $\theta$ has to be in a convex and compact set; 2) if $\theta_{s,a}$ is too large, there will be numerical issues when running the algorithm. It is known that the softmax function is Lipschitz continuous \cite{gao2018propertiessoftmaxfunctionapplication}. And the Hessian of the softmax function is bounded \cite{wei2025activeinferenceincentivedesign}. Thus, it satisfies the assumption~\ref{assume:bounded-policy-gradient}.

\paragraph{Last-state opacity}
The optimization process incorporates the constraint that the total return must be $\zeta \ge 0.3$, with a time horizon of $T = 12$. For optimizing last-state opacity, the possible initial states are the cells in the first column of the grid world, with the initial-state distribution following a discrete uniform distribution over these states. Fig.~\ref{fig:last_opacity_results} presents the estimated values of last-state opacity and the total return obtained using the primal-dual policy gradient method.  
We rely on the estimated total return, rather than the exact total return from value iteration, because computing the latter can be computationally expensive.

\begin{figure}[t]
\includegraphics[width=\linewidth]{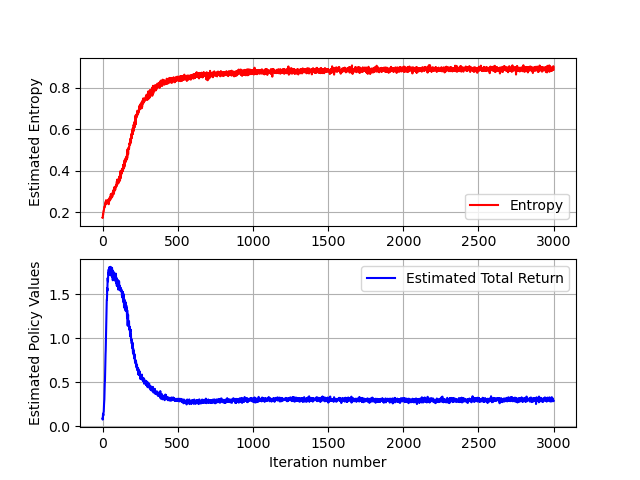}
\centering
\caption{The result of the primal-dual policy gradient method for enforcing optimal last-state opacity under the constraint on the total return.}
\label{fig:last_opacity_results}
\end{figure}

Fig.~\ref{fig:last_opacity_results} illustrates when the algorithm converges, the conditional entropy $H(Z_T|Y)$ eventually approaches $0.888$. The value of policy reaches $0.313$, satisfying the predefined threshold of $\zeta = 0.3$. 
Given the conditional entropy is close to $1$, which is the maximal value of the entropy, the observation reveals little information regarding whether a secret location is visited, even when the observer knows the exact policy used by the robot. 

There is generally an inverse relationship between the conditional entropy and the total return: higher returns are typically associated with lower conditional entropy.
This behavior is a consequence of the environmental configuration. From observing the sampled trajectories (e.g., $(0, 5) \to (1, 5) \to (2, 5) \to (2, 4) \to (3, 4) \to (4, 4) \to (5, 4) \to (5, 3) \to (4, 3) \to (5, 3) \to (4, 3) \to (4, 2) \to (4, 2)$), we noticed that the agent navigates among various sensor ranges or avoids the sensor ranges to confuse the observer, thereby diminishing the total return that P1 can achieve, as P1 is constrained from remaining stationary at the goal. 

\paragraph{Initial-state opacity}
Next, we use the same example to demonstrate the policy that maximizes the initial state opacity under the same constraint on the total return and horizon length.  In this case, the initial-state distribution follows a discrete uniform distribution over a set of possible initial states  $(0,0)$, $(0,5)$, $(5,5)$. The goal of the robot is to make the observer uncertain about which initial state from this set is selected.
Fig.~\ref{fig:initial_opacity_results} illustrates the initial-state opacity, measured by conditional entropy $H(S_0|Y;\theta)$, and the value of the policy $\pi_\theta$ within the primal-dual policy gradient method.
\begin{figure}[t]
\includegraphics[width=\linewidth]{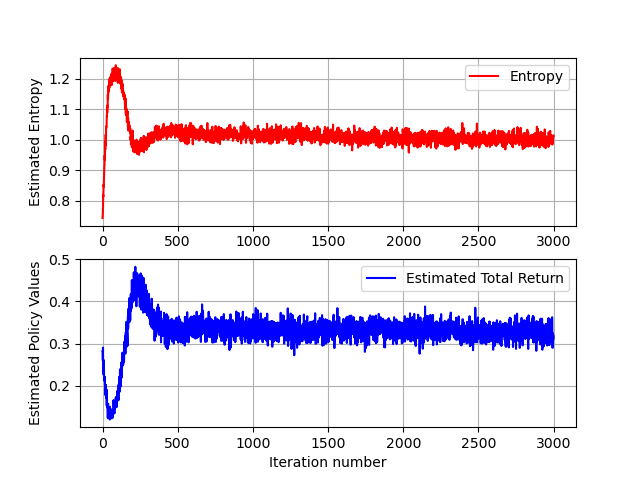}
\centering
\caption{The result of the primal-dual policy gradient method for enforcing optimal initial-state opacity under the constraint on the total return.}
\label{fig:initial_opacity_results}
\end{figure}
When the algorithm converges, the initial conditional entropy eventually reaches $1.01$ while the maximal entropy is $1.58$ in this case. The policy value reaches $0.311$, satisfying the predefined threshold of $\zeta= 0.3$. 

In this environment, elevated initial-state opacity indicates the observer is more uncertain about P1's specific initial state.
The entropy is initially very high at the start of optimization because a randomized policy is used as the initial policy, and none of the initial states are within the sensor range. However, to reach the goal, the agent must sacrifice some initial state uncertainty. This is because, while moving toward the goal, the agent inevitably enters some sensor ranges, which exposes information about its initial state.

\subsection{Graph Example for Language-Based Opacity} 
We illustrate the effectiveness of the proposed language-based opacity-enforcing planning algorithm with a small MDP.
The graph of the \ac{mdp}, shown in Fig.~\ref{fig:graph_example}, consists of seven states and two actions $\{a_1, a_2\}$. The initial state is $0$.   The reward function $R(s, a) = 1$ if $s = 2$; $R(s, a) = 0.1$ if $s = 4$, for any $a\in A$ and $R(s,a)=0$ for any other state-action pairs. 
We set the lower bound on the total return to be $\zeta = 0.2$ and the planning horizon to be $T = 5$. Since states $\{2,4\}$ are not sink states, the agent is allowed to repeatedly visit the goal states. 

Let $\calAP=\{p_s, p_g\}$ where $p_s$ means a secret state is visited and $p_g$ means a goal state is visited.
We define the label function of the LMDP as $L(s) = \{p_s\}$ if $s = h_1, h_2$; $L(s) = \{p_g\}$ if $s = 2, 4$, and $L(s) = \varnothing$ otherwise. The   LTL$_f$ formula is  $\varphi = \Eventually (p_s \land \Eventually p_g) \lor \Eventually p_g$ where $\Eventually$ is the temporal logic operator for ``eventually''. In plain language, the formula can be satisfied by a path that visits a secret and then a goal state, or just visits the goal state. The formula
 can be represented by a $4$-state automaton shown in Fig.~\ref{fig:automaton}. The agent aims to make the observer uncertain about the progress in satisfying the task within each finite run. For example, if the agent reaches $q_3$ but the observer believes the state can be $q_0,q_1,q_2,q_3$ with similar probabilities, then the agent protect the secrecy in the property being satisfied.

\begin{figure}
\label{fig:automaton}
\begin{tikzpicture}[shorten >=1pt,node distance=2.5cm,on grid,auto] 
    \node[state, initial] (0)   {$q_0$}; 
    \node[state,accepting] (1) [above right=of 0] {$q_1$};    
    \node[state] (2) [right=of 0] {$q_2$};
    \node[state] (3) [right=of 2,accepting] {$q_3$};  
\path[->] 
    (0) edge   node{$\{p_g\}$} (1)
        (0) edge  node{$\{p_s\}$} (2)
                 
(2) edge   node{$\{p_g\}$} (3)
(2) edge [loop above] node{$\{p_s\} $} (2) 
(3) edge [loop right]  node{$\truev$} (3)
(1) edge [loop right]  node{$\truev$} (1)
;
\end{tikzpicture}
    \caption{The \ac{dfa} for the secret task. Self-loops with label $\varnothing$ are omitted.}
\end{figure}
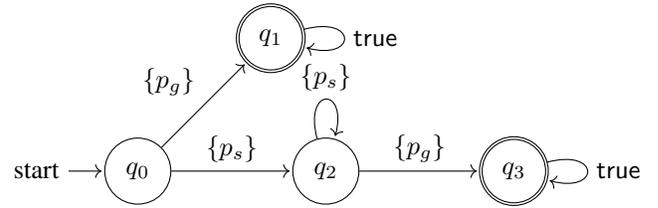
 At   states $0,h_1, 2$, the observer has full observation. At states $1,3,4,h_2$, the observer has partial observation, described as follows: The observer is equipped with two sensors with distinct ranges indicated by the red states and blue states. The observer randomly selects a sensor to query at each step. As the agent enters the senor range $\{1,4\}$ (resp. $\{3,h_2\}$), the observer receives   observation  ``r" (resp. ``b").
\begin{figure}[t]
\centering
\begin{tikzpicture}[->,>=stealth',shorten >=1pt,auto,node distance=2.5cm, scale=0.8,transform shape]
	\node[state] (0)  {\Large $0$};
	\node[state, fill=red!60] (1) [right=1cm of 0] {\Large $1$};
	\node[state, double, fill=red!60] (4) [below=2cm of 1] {\Large $4$};
 	\node[state] (5) [above=1cm of 1] {\Large $h_1$};
	\node[state, double] (2) [right=4cm of 5] {\Large $2$};
	\node[state, fill=blue!15] (3) [right=4cm of 1] {\Large $3$};
	\node[state, fill=blue!15] (6)  [right=4cm of 4] {\Large $h_2$};
	
	\path 	
		(0) edge   node {$a_1$} (5)
		(0) edge   node {$a_2$} (1)
		(5) edge[bend left]   node {$a_1, a_2$} (2)
		(1) edge[bend left]   node[above] {$a_1:0.5$} (3)
            (1) edge   node {$a_1:0.5$} (4)
            (1) edge   node[above, pos=0.2] {$\quad a_2$} (6)
		(2) edge[bend left]   node { $a_1$} (3)
            (2) edge   node {$a_2$} (5)
            (3) edge   node {$a_1:0.5$} (1)
            (3) edge[bend left]   node[right] {$a_1:0.5$} (6)
            (3) edge   node {$a_2$} (2)
		(6) edge   node {$a_1, a_2: 0.5$} (3)
            (6) edge[bend left]   node {$a_1, a_2: 0.5$} (4)
            (4) edge[bend left]   node {$a_1:0.5$} (1)
		(4) edge   node {$a_1:0.5$} (6)
            (4) edge   node[below, pos=0.2] {$a_2$} (3);
\end{tikzpicture}
\caption{An illustrative example of opacity-enforcing winning. The arrows with $a_1$ or $a_2$ represent a deterministic action, the arrows with $a_1: 0.5$ represent a stochastic action with probability $0.5$ to a certain state, and the arrows with $a_1, a_2: 0.5$ represent that the agent will transfer to a certain state with probability $0.5$ by both actions. Double circle nodes represent goal states.}
\label{fig:graph_example}
\end{figure}
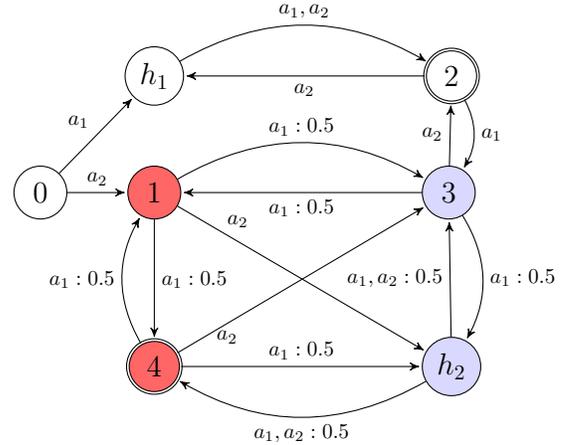

We will adopt the soft-max policy parameterization, represented by the equation:
\[
\pi_{\theta}(a|v) = \frac{\exp(\theta_{v,a})}{\sum_{a' \in \mathcal{A}}\exp(\theta_{v, a'})},
\]
where $\theta \in \mathbb{R}^{|\mathcal{V} \times \mathcal{A}|}$ denotes the policy parameter vector.

Fig.~\ref{fig:opacity_results} shows the estimated values of language-based opacity $H(Q_T|Y,\theta_t)$ alongside the total return $V(\mu_0,\theta_t)$ obtained through the primal-dual policy gradient method. Notably, we rely on the estimated total return rather than the analytically derived total return from value iterations. Because we employ the gradient estimate of the total return in the proposed policy gradient method.
\begin{figure}[t]
\includegraphics[width=\linewidth]{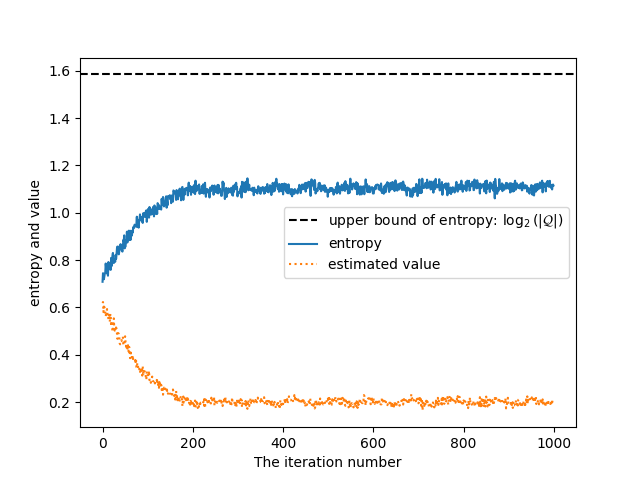}
\centering
\caption{The result of the primal-dual policy gradient algorithm. The solid blue line represents the opacity and the dotted orange line represents the estimated total return.}
\label{fig:opacity_results}
\end{figure}
Fig.~\ref{fig:opacity_results} illustrates when the algorithm converges, the conditional entropy $H(Q_T|Y,\theta_t)$ eventually approaches $1.115$ as $t$ approaches $1000$. The policy value reaches $0.207$, satisfying the predefined threshold of $\zeta = 0.2$. The conditional entropy $0 \le H(Q_T|Y, \theta) \le \log(|\mathcal{Q}| - 1)$ for any policy parameter $\theta$. The upper bound $\log(|\mathcal{Q}| - 1) = 1.585$ in this example. The conditional entropy is inversely proportional to the value/total return. By observing the sampled trajectories, we noticed that the agent goes to the states with sensors to obfuscate the visited states, thereby diminishing the total return that P1 can achieve since  to protect secrecy for the temporal logic formula, the agent can only visit state 4 and obtain a lower reward.

Given the absence of alternative algorithms tailored specifically for addressing the proposed opacity-enforcement planning problems, we opt for a comparative analysis against a baseline algorithm designed for entropy-regularized MDPs. 
In this baseline framework, the objective value is formulated as a weighted sum of the total return and the discounted entropy of the policy \cite{Nachum2017bridge}. Notably, as the weight assigned to the entropy terms increases, the resulting policy tends to become more stochastic and random.

The objective function for entropy-regularized \ac{mdp}s is given by \cite{Nachum2017bridge, Cen2022fastentropy}:
\begin{equation}
V_\tau(s;\theta) = V(s;\theta) + \tau H(s;\theta).
\end{equation}
The discounted entropy term, $H(s;\theta)$, is defined as:
\begin{equation}
H(s;\theta) = -\frac{1}{1 - \gamma} E_{s \sim d_{\pi_\theta}} \Big[ \sum_{a \in \mathcal{A}} \pi_\theta(a|s) \log \pi_\theta(a|s) \Big],
\end{equation}
where $d_{\pi_\theta}$ is the occupancy measure induced by policy $\pi_\theta$. The objective is to maximize the regularized total return $V_\tau(s;\theta)$.

To facilitate the comparison between the two methods, we introduce the metric $P_E$ (probability of guess error), defined as:
\begin{equation}
\label{eq:prob_error}
P_E \coloneqq P_\theta(Q_T \neq \hat{q}_T \mid y),
\end{equation}
where $\hat{q}_T$ represents the maximum likelihood estimator of $Q_T$, given by $\hat{q}_T \coloneqq \arg\max_{q_T} P_\theta(q_T|y)$. Indeed, a higher value of $P_E$ signifies that the observer is more likely to make an incorrect guess regarding the state of the agent. Consequently, the agent can conceal the secret more effectively under such circumstances. Therefore, a method yielding a higher $P_E$ can be considered more effective in preserving secrecy.

We compare our method with entropy-regularized MDPs solved with different $\tau$ values, selecting $11$ values from $\tau = 1$ to $\tau = 20$. The following graphs (Fig.~\ref {fig:guess_error}) illustrate the results. Note that when $\tau \approx 5$, the policy approaches a random policy, making it impractical to further increase opacity by raising $\tau$. 
\begin{figure}[t]
\includegraphics[width=\linewidth]{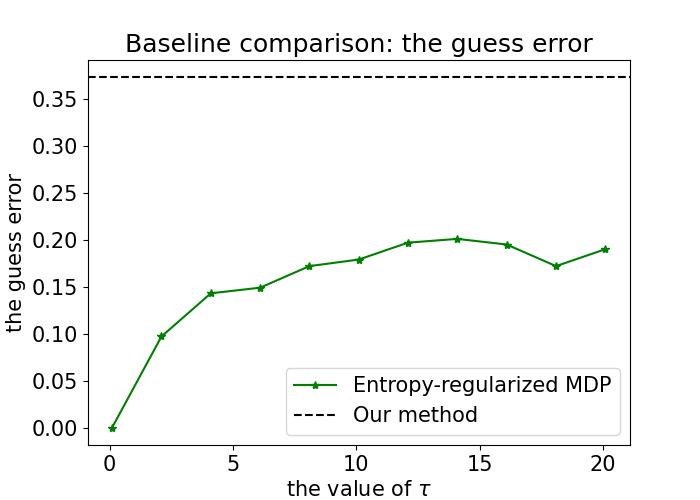}
\centering
\caption{Comparison with baseline (entropy regularized MDP). The dashed line is the probability of guess error from our method. The green line is the probability of guess error from the baseline method.}
\label{fig:guess_error}
\end{figure}
The results depicted in Fig.~\ref{fig:guess_error} underscore the inability of standard entropy-regularized MDP to achieve maximized language-based opacity. Specifically, the observed $P_E$ values only reach approximately $0.2$ when utilizing the regularized MDP approach. In contrast, our proposed method yields significantly higher probabilities of error, reaching $0.384$.

These comparative experiments highlight the capability of our proposed methods to leverage the noise inherent in observations to optimize opacity. Such a feature cannot be attained through policies derived from entropy-regularized MDPs, which do not incorporate the observation function.

\section{Conclusion}

This paper introduces a framework for enforcing information-theoretic opacity in systems modeled as Markov Decision Processes (MDPs), addressing both state-based and language-based secrets. The paper proposes measuring information leakage using the conditional entropy of the secret given the observer's partial observations, as a higher conditional entropy makes it more challenging for an observer to infer the secret from partial observations.

The primary contribution of our work is the development of a novel methodology to maximize information-theoretic opacity in an MDP, subject to task performance constraints. This method provides a strong notion of opacity due to the assumption of an informed observer who is aware of the control policy and system dynamics. The paper demonstrates that finite-memory policies can outperform Markov policies in achieving higher opacity. 
Then, it presents algorithms to solve constrained opacity-enforcement problems for initial-state, last-state, and language-based opacity.
To overcome the non-cumulative nature of conditional entropy, which prevents the use of standard POMDP solvers, a novel method based on observable operators is introduced to compute the gradient of conditional entropy with respect to policy parameters. Utilizing this gradient, a primal-dual gradient-based algorithm is developed and proven to converge to a locally optimal, feasible policy that maximizes opacity while satisfying a total return constraint. A key take-home message is that these proposed methods effectively leverage the inherent noise in observations to significantly increase uncertainty for the observer regarding secret information, as validated by experimental results.

Despite these advances, challenges remain,  particularly in scaling the approach to more complex scenarios. The exact computation of the gradient of conditional entropy requires summing over all possible observation sequences, which is combinatorially large and thus computationally expensive for longer horizons or larger observation spaces.  

Future research could explore several directions. First,    more computationally efficient methods for gradient estimation in large-scale systems would be beneficial. Second, the current framework assumes an informed observer; extending the method to scenarios with less informed observers or those with uncertainties about the system model or policy could yield both theoretical insights and practical applications. Finally, generalizing the information-theoretic opacity concepts to cyber-physical systems with continuous dynamics would enhance their relevance for real-world applications.

\section*{References}
\vspace*{-2em}
\bibliographystyle{IEEEtranS}
\bibliography{chongyang_refs, Chongyang_refs_2}

\begin{thebibliography}{10}
\providecommand{\url}[1]{#1}
\csname url@samestyle\endcsname
\providecommand{\newblock}{\relax}
\providecommand{\bibinfo}[2]{#2}
\providecommand{\BIBentrySTDinterwordspacing}{\spaceskip=0pt\relax}
\providecommand{\BIBentryALTinterwordstretchfactor}{4}
\providecommand{\BIBentryALTinterwordspacing}{\spaceskip=\fontdimen2\font plus
\BIBentryALTinterwordstretchfactor\fontdimen3\font minus \fontdimen4\font\relax}
\providecommand{\BIBforeignlanguage}[2]{{%
\expandafter\ifx\csname l@#1\endcsname\relax
\typeout{** WARNING: IEEEtranS.bst: No hyphenation pattern has been}%
\typeout{** loaded for the language `#1'. Using the pattern for}%
\typeout{** the default language instead.}%
\else
\language=\csname l@#1\endcsname
\fi
#2}}
\providecommand{\BIBdecl}{\relax}
\BIBdecl

\bibitem{baum1970maximization}
L.~E. Baum, T.~Petrie, G.~Soules, and N.~Weiss, ``A maximization technique occurring in the statistical analysis of probabilistic functions of markov chains,'' \emph{The annals of mathematical statistics}, vol.~41, no.~1, pp. 164--171, 1970.

\bibitem{Borkar2008}
V.~Borkar, ``Stochastic approximation. a dynamical systems viewpoint,'' 01 2008.

\bibitem{bryansOpacityGeneralisedTransition2006}
J.~W. Bryans, M.~Koutny, L.~Mazaré, and P.~Y.~A. Ryan, ``\BIBforeignlanguage{en}{Opacity {Generalised} to {Transition} {Systems}},'' in \emph{\BIBforeignlanguage{en}{Formal {Aspects} in {Security} and {Trust}}}, ser. Lecture {Notes} in {Computer} {Science}, T.~Dimitrakos, F.~Martinelli, P.~Y.~A. Ryan, and S.~Schneider, Eds.\hskip 1em plus 0.5em minus 0.4em\relax Berlin, Heidelberg: Springer, 2006, pp. 81--95.

\bibitem{berardProbabilisticOpacityMarkov2015}
\BIBentryALTinterwordspacing
B.~Bérard, K.~Chatterjee, and N.~Sznajder, ``Probabilistic opacity for {Markov} decision processes,'' \emph{Information Processing Letters}, vol. 115, no.~1, pp. 52--59, Jan. 2015. [Online]. Available: \url{https://doi.org/10.1016/j.ipl.2014.09.001}
\BIBentrySTDinterwordspacing

\bibitem{berardQuantifyingOpacity2015}
\BIBentryALTinterwordspacing
B.~Bérard, J.~Mullins, and M.~Sassolas, ``\BIBforeignlanguage{en}{Quantifying opacity†},'' \emph{\BIBforeignlanguage{en}{Mathematical Structures in Computer Science}}, vol.~25, no.~2, pp. 361--403, Feb. 2015, publisher: Cambridge University Press. [Online]. Available: \url{https://www.cambridge.org/core/journals/mathematical-structures-in-computer-science/article/quantifying-opacity/6B982B414989B89BC7B0A74358B04AEB}
\BIBentrySTDinterwordspacing

\bibitem{Cen2022fastentropy}
\BIBentryALTinterwordspacing
S.~Cen, C.~Cheng, Y.~Chen, Y.~Wei, and Y.~Chi, ``Fast global convergence of natural policy gradient methods with entropy regularization,'' \emph{Operations Research}, vol.~70, no.~4, pp. 2563--2578, 2022. [Online]. Available: \url{https://doi.org/10.1287/opre.2021.2151}
\BIBentrySTDinterwordspacing

\bibitem{Chen2023symbolic}
\BIBentryALTinterwordspacing
B.~Chen, K.~Leahy, A.~Jones, and M.~Hale, ``Differential privacy for symbolic systems with application to markov chains,'' \emph{Automatica}, vol. 152, no.~C, jun 2023. [Online]. Available: \url{https://doi.org/10.1016/j.automatica.2023.110908}
\BIBentrySTDinterwordspacing

\bibitem{Wiley2005information}
T.~M. Cover and J.~A. Thomas, \emph{Elements of Information Theory (Wiley Series in Telecommunications and Signal Processing)}.\hskip 1em plus 0.5em minus 0.4em\relax USA: Wiley-Interscience, 2006.

\bibitem{de2013linear}
G.~De~Giacomo and M.~Y. Vardi, ``Linear temporal logic and linear dynamic logic on finite traces,'' in \emph{IJCAI'13 Proceedings of the Twenty-Third international joint conference on Artificial Intelligence}.\hskip 1em plus 0.5em minus 0.4em\relax Association for Computing Machinery, 2013, pp. 854--860.

\bibitem{dubreil2008opacity}
J.~Dubreil, P.~Darondeau, and H.~Marchand, ``Opacity enforcing control synthesis,'' in \emph{2008 9th International Workshop on Discrete Event Systems}, 2008, pp. 28--35.

\bibitem{gao2018propertiessoftmaxfunctionapplication}
\BIBentryALTinterwordspacing
B.~Gao and L.~Pavel, ``On the properties of the softmax function with application in game theory and reinforcement learning,'' 2018. [Online]. Available: \url{https://arxiv.org/abs/1704.00805}
\BIBentrySTDinterwordspacing

\bibitem{HanX2023Scai}
X.~Han, K.~Zhang, J.~Zhang, Z.~Li, and Z.~Chen, ``\BIBforeignlanguage{eng}{Strong current-state and initial-state opacity of discrete-event systems},'' \emph{\BIBforeignlanguage{eng}{Automatica (Oxford)}}, vol. 148, p. 110756, 2023.

\bibitem{jaeger2000observableoperator}
H.~Jaeger, ``{Observable Operator Models for Discrete Stochastic Time Series},'' \emph{Neural Computation}, vol.~12, no.~6, pp. 1371--1398, 06 2000.

\bibitem{keroglouProbabilisticSystemOpacity2016}
C.~Keroglou and C.~N. Hadjicostis, ``Probabilistic system opacity in discrete event systems,'' in \emph{2016 13th {International} {Workshop} on {Discrete} {Event} {Systems} ({WODES})}, May 2016, pp. 379--384.

\bibitem{khouzani2017leakage}
M.~Khouzani and P.~Malacaria, ``Leakage-minimal design: Universality, limitations, and applications,'' in \emph{2017 IEEE 30th Computer Security Foundations Symposium (CSF)}, 2017, pp. 305--317.

\bibitem{Lee2019avoid}
\BIBentryALTinterwordspacing
J.~D. Lee, I.~Panageas, G.~Piliouras, M.~Simchowitz, M.~I. Jordan, and B.~Recht, ``First-order methods almost always avoid strict saddle points,'' \emph{Math. Program.}, vol. 176, no. 1–2, p. 311–337, Jul. 2019. [Online]. Available: \url{https://doi.org/10.1007/s10107-019-01374-3}
\BIBentrySTDinterwordspacing

\bibitem{lin2011opacity}
F.~Lin, ``Opacity of discrete event systems and its applications,'' \emph{Automatica}, vol.~47, no.~3, pp. 496--503, 2011.

\bibitem{Mazar2004UsingUF}
\BIBentryALTinterwordspacing
L.~Mazar{\'e}, ``Using unification for opacity properties,'' 2004. [Online]. Available: \url{https://api.semanticscholar.org/CorpusID:54920181}
\BIBentrySTDinterwordspacing

\bibitem{mirInformationTheoreticFoundationsDifferential2013}
D.~J. Mir, ``\BIBforeignlanguage{en}{Information-{Theoretic} {Foundations} of {Differential} {Privacy}},'' in \emph{\BIBforeignlanguage{en}{Foundations and {Practice} of {Security}}}, J.~Garcia-Alfaro, F.~Cuppens, N.~Cuppens-Boulahia, A.~Miri, and N.~Tawbi, Eds.\hskip 1em plus 0.5em minus 0.4em\relax Berlin, Heidelberg: Springer, 2013, pp. 374--381.

\bibitem{molloy2023smoother}
T.~Molloy and G.~Nair, ``Smoother entropy for active state trajectory estimation and obfuscation in pomdps,'' \emph{IEEE Transactions on Automatic Control}, vol.~PP, pp. 1--16, 06 2023.

\bibitem{Nachum2017bridge}
\BIBentryALTinterwordspacing
O.~Nachum, M.~Norouzi, K.~Xu, and D.~Schuurmans, ``Bridging the gap between value and policy based reinforcement learning,'' 2017. [Online]. Available: \url{https://arxiv.org/pdf/1702.08892.pdf}
\BIBentrySTDinterwordspacing

\bibitem{saboori2007notions}
A.~Saboori and C.~N. Hadjicostis, ``Notions of security and opacity in discrete event systems,'' in \emph{2007 46th IEEE Conference on Decision and Control}.\hskip 1em plus 0.5em minus 0.4em\relax IEEE, 2007, pp. 5056--5061.

\bibitem{Saboori2009Kstep}
------, ``Verification of k-step opacity and analysis of its complexity,'' in \emph{Proceedings of the 48h IEEE Conference on Decision and Control (CDC) held jointly with 2009 28th Chinese Control Conference}, 2009, pp. 205--210.

\bibitem{saboori2013current}
------, ``Current-state opacity formulations in probabilistic finite automata,'' \emph{IEEE Transactions on automatic control}, vol.~59, no.~1, pp. 120--133, 2013.

\bibitem{saboori_verification_2013}
\BIBentryALTinterwordspacing
------, ``Verification of initial-state opacity in security applications of discrete event systems,'' \emph{Information Sciences}, vol. 246, pp. 115--132, 2013. [Online]. Available: \url{https://www.sciencedirect.com/science/article/pii/S0020025513004143}
\BIBentrySTDinterwordspacing

\bibitem{savas2020temporal}
Y.~Savas, M.~Ornik, M.~Cubuktepe, M.~O. Karabag, and U.~Topcu, ``Entropy maximization for markov decision processes under temporal logic constraints,'' \emph{IEEE Transactions on Automatic Control}, vol.~65, no.~4, pp. 1552--1567, 2020.

\bibitem{shannonCommunicationTheorySecrecy1949}
C.~E. Shannon, ``Communication theory of secrecy systems,'' \emph{The Bell System Technical Journal}, vol.~28, no.~4, pp. 656--715, 1949.

\bibitem{shi2023synthesis}
C.~Shi, A.~N. Kulkarni, H.~Rahmani, and J.~Fu, ``Synthesis of opacity-enforcing winning strategies against colluded opponent,'' 2023.

\bibitem{SHU20083054}
\BIBentryALTinterwordspacing
S.~Shu, F.~Lin, H.~Ying, and X.~Chen, ``State estimation and detectability of probabilistic discrete event systems,'' \emph{Automatica}, vol.~44, no.~12, pp. 3054--3060, 2008. [Online]. Available: \url{https://www.sciencedirect.com/science/article/pii/S0005109808003257}
\BIBentrySTDinterwordspacing

\bibitem{Sutton1999policy}
\BIBentryALTinterwordspacing
R.~S. Sutton, D.~McAllester, S.~Singh, and Y.~Mansour, ``Policy gradient methods for reinforcement learning with function approximation,'' in \emph{Advances in Neural Information Processing Systems}, S.~Solla, T.~Leen, and K.~M\"{u}ller, Eds., vol.~12.\hskip 1em plus 0.5em minus 0.4em\relax MIT Press, 1999. [Online]. Available: \url{https://proceedings.neurips.cc/paper_files/paper/1999/file/464d828b85b0bed98e80ade0a5c43b0f-Paper.pdf}
\BIBentrySTDinterwordspacing

\bibitem{Shie2018rewardconstraint}
\BIBentryALTinterwordspacing
C.~Tessler, D.~J. Mankowitz, and S.~Mannor, ``Reward constrained policy optimization,'' \emph{CoRR}, vol. abs/1805.11074, 2018. [Online]. Available: \url{http://arxiv.org/abs/1805.11074}
\BIBentrySTDinterwordspacing

\bibitem{udupa2025}
\BIBentryALTinterwordspacing
S.~Udupa, C.~Shi, and J.~Fu, ``Synthesis of dynamic masks for information-theoretic opacity in stochastic systems,'' 2025. [Online]. Available: \url{https://arxiv.org/abs/2502.10552}
\BIBentrySTDinterwordspacing

\bibitem{vardi2005automata}
M.~Y. Vardi, ``An automata-theoretic approach to linear temporal logic,'' \emph{Logics for concurrency: structure versus automata}, pp. 238--266, 2005.

\bibitem{Parv2013Privacy}
P.~Venkitasubramaniam, ``Privacy in stochastic control: A markov decision process perspective,'' in \emph{2013 51st Annual Allerton Conference on Communication, Control, and Computing (Allerton)}, 2013, pp. 381--388.

\bibitem{wei2025activeinferenceincentivedesign}
\BIBentryALTinterwordspacing
X.~Wei, C.~Shi, S.~Han, A.~H. Hemida, C.~A. Kamhoua, and J.~Fu, ``Active inference through incentive design in markov decision processes,'' 2025. [Online]. Available: \url{https://arxiv.org/abs/2502.07065}
\BIBentrySTDinterwordspacing

\bibitem{wintenberg2022general}
A.~Wintenberg, M.~Blischke, S.~Lafortune, and N.~Ozay, ``A general language-based framework for specifying and verifying notions of opacity,'' \emph{Discrete Event Dynamic Systems}, vol.~32, no.~2, pp. 253--289, 2022.

\bibitem{yasuokaQuantitativeInformationFlow2010}
\BIBentryALTinterwordspacing
H.~Yasuoka and T.~Terauchi, ``Quantitative information flow - verification hardness and possibilities,'' in \emph{2010 IEEE 23rd Computer Security Foundations Symposium (CSF 2010)}.\hskip 1em plus 0.5em minus 0.4em\relax Los Alamitos, CA, USA: IEEE Computer Society, jul 2010, pp. 15--27. [Online]. Available: \url{https://doi.ieeecomputersociety.org/10.1109/CSF.2010.9}
\BIBentrySTDinterwordspacing

\bibitem{yinInfinitestepOpacityKstep2019}
\BIBentryALTinterwordspacing
X.~Yin, Z.~Li, W.~Wang, and S.~Li, ``\BIBforeignlanguage{en}{Infinite-step opacity and {K}-step opacity of stochastic discrete-event systems},'' \emph{\BIBforeignlanguage{en}{Automatica}}, vol.~99, pp. 266--274, Jan. 2019. [Online]. Available: \url{https://www.sciencedirect.com/science/article/pii/S0005109818305235}
\BIBentrySTDinterwordspacing

\bibitem{yin2019approximate}
X.~Yin and M.~Zamani, ``On approximate opacity of cyber-physical systems,'' 02 2019.

\bibitem{Yin2021approximate}
X.~Yin, M.~Zamani, and S.~Liu, ``On approximate opacity of cyber-physical systems,'' \emph{IEEE Transactions on Automatic Control}, vol.~66, no.~4, pp. 1630--1645, 2021.

\end{thebibliography}

\appendix
\subsection{Compute a finite-memory policy for opacity enforcement} 
\label{app:finite_memory_reduce}
A finite-memory policy is defined as a tuple $\pi \coloneqq \langle \mathfrak{M}, \calS, \calA, \delta, \psi, m_0 \rangle$
where 
 $\mathfrak{M}$ is a finite set of memory states, 
$\mathcal{S}$ is the set of environment states, 
$\mathcal{A}$ is the set of actions, 
$\delta : \mathfrak{M} \times \mathcal{S} \times \mathcal{A} \rightarrow \mathrm{Dist}(\mathfrak{M})$ 
is the memory update function specifying the distribution over next memory states, and 
$\psi : \mathfrak{M}  \rightarrow \mathrm{Dist}(\mathcal{A})$ 
is the action selection function, 
and $m_0 \in \mathfrak{M}$ is the initial memory state.

 To solve a finite-memory policy, we can reduce it to solving a Markov policy in the augmented-state Markov decision process defined as follows:
\[
\hat M =\langle \calS\times \mathfrak{M},  \calA, \hat P, \hat \mu_0, \hat R \rangle
\]
where $\calS \times \mathfrak{M}  $ is the augmented state space, $\calA$ is the action space. $\hat P$ is defined by
$
\hat P((s',m')|(s,m),a) = P(s'|s,a)\delta(m'|m, (s,a )) $ and the initial state distribution is
$\hat\mu_0(s,m) =\mu_0(s)$
where $m= \delta(m_0, s)$. The reward function is defined such that 
\[
\hat{R}((s,m),a)= R(s,a), \forall (s,a)\in \calS\times \calA.
\]

In this augmented-state MDP, a Markov policy $\pi: S\times \mathfrak{M} \rightarrow \Delta(\calA)$ is a finite-memory policy in the original MDP. Thus, by solving a constrained opacity-enforcement policy in the augmented-state MDP, we recover a finite-memory policy in the original MDP.
\normalcolor
\subsection{Derivations of the sample approximation}
\label{app:approximation}
First, we   extract probability $P_\theta(y)$ from equation~\eqref{eq:HMM_entropy} as a weight, then we obtain
\begin{equation}
H(Z_T|Y;\theta) = - \sum_{y \in \mathcal{O}^T} P_\theta(y) \sum_{z_T \in \{0,1\}} P_\theta(z_T| y) \log P_\theta(z_T | y). 
\end{equation}
Thus, we can approximate the conditional entropy by equation~\eqref{eq:HMM_approx_entropy}. Similarly,
\begin{equation}
\begin{aligned}
 &\nabla_\theta H(Z_T|Y;\theta) \\
= & - \sum_{y \in \mathcal{O}^T} \sum_{z_T \in \{0,1\}} \Big[\nabla_\theta P_\theta(z_T, y) \log P_\theta(z_T | y) \\
&+  P_\theta(z_T, y) \nabla_\theta  \log P_\theta(z_T | y)\Big] \\
= & - \sum_{y \in \mathcal{O}^T} \sum_{z_T \in \{0,1\}} \Big[P_\theta(y) \nabla_\theta P_\theta(z_T| y) \log P_\theta(z_T | y) + \\ 
& P_\theta(z_T| y) \nabla_\theta P_\theta(y) \log P_\theta(z_T | y) 
+   P_\theta(y)\frac{\nabla_\theta P_\theta(z_T | y)}{\ln 2}\Big] \\
 = & - \sum_{y \in \mathcal{O}^T} P_\theta (y) \sum_{z_T \in \{0,1\}} \Big[ \log P_\theta(z_T | y) \nabla_\theta P_\theta(z_T| y) \\
& + P_\theta(z_T| y) \log P_\theta(z_T | y) \nabla_\theta \log P_\theta(y) + \frac{\nabla_\theta P_\theta(z_T | y)}{\ln 2} \Big].
\end{aligned}
\end{equation}
Then we can approximate the gradient by equation~\eqref{eq:HMM_approx_gradient_entropy}.

\subsection{Computing the gradient and hessian of the   $P_\theta(z_T|y)$}
Based on  ~\eqref{eq:HMM_gradient_P_zT_y}, we obtain
\begin{equation}
\label{eq:hessian_p_z_y}
\begin{aligned}
\nabla_\theta^2 & \log P_\theta(z_T|y) = \frac{P_\theta(z_T|y) \nabla_\theta^2 P_\theta(z_T|y)}{P_\theta^2(z_T|y)} \\
&- \frac{\nabla_\theta P_\theta(z_T|y) \nabla_\theta P_\theta(z_T|y)^\top}{P_\theta^2(z_T|y)} \\
& = \frac{\nabla_\theta^2 P_\theta(z_T|y)}{P_\theta(z_T|y)}  - \frac{\nabla_\theta P_\theta(z_T|y) \nabla_\theta P_\theta(z_T|y)^\top}{P_\theta^2(z_T|y)}.
\end{aligned}
\end{equation}
The Hessian is given by
\begin{equation}
\begin{aligned}
\nabla_\theta^2 P_\theta(z_T |y)   = & \sum_{s_{T} \in W} \Big\{ \nabla_\theta \Big[\frac{\nabla_\theta  P_\theta(s_{T}, y)}{P_\theta(y)}\Big] \\
& - \nabla_\theta\Big[\frac{ P_\theta(s_{T}, y)}{ P_\theta^2(y)} \nabla_\theta P_\theta(y) \Big] \Big\}.
\end{aligned}
\end{equation}
The gradient of the first term is
\begin{equation}
\label{eq:term1}
\begin{aligned}
\nabla_\theta & \Big[ \frac{\nabla_\theta P_\theta(s_{T}, y)}{P_\theta(y)} \Big] = \frac{P_\theta(y) \nabla_\theta^2 P_\theta(s_{T}, y)}{P_\theta^2(y)}  \\
& - \frac{\nabla_\theta P_\theta(s_{T}, y) \nabla_\theta P_\theta(y)^\top}{P_\theta^2(y)} \\
& = \frac{\nabla_\theta^2 P_\theta(s_{T}, y)}{P_\theta(y)} - \frac{\nabla_\theta P_\theta(s_{T}, y) \nabla_\theta P_\theta(y)^\top}{P_\theta^2(y)}
\end{aligned}
\end{equation}
and the gradient of the second term is
\begin{equation}
\label{eq:term2}
\begin{aligned}
& \nabla_\theta \Big[ \frac{ P_\theta(s_{T}, y)}{ P_\theta^2(y)} \nabla_\theta P_\theta(y) \Big] = \frac{P_\theta^2(y)\nabla_\theta P_\theta(y) \nabla_\theta P_\theta(s_{T}, y)^\top}{P_\theta^4(y)} \\
& + \frac{P_\theta^2(y) P_\theta(s_{T}, y) \nabla_\theta^2 P_\theta(y)}{P_\theta^4(y)} - + \frac{P_\theta(s_{T}, y) \nabla_\theta P_\theta(y) \nabla_\theta P_\theta^2(y)^\top}{P_\theta^4(y)} \\
&  = \frac{\nabla_\theta P_\theta(y) \nabla_\theta P_\theta(s_{T}, y)^\top}{P_\theta^2(y)} + \frac{P_\theta(s_{T}, y) \nabla_\theta^2 P_\theta(y)}{P_\theta^2(y)} \\ 
& - \frac{2 P_\theta(s_{T}, y) \nabla_\theta P_\theta(y) \nabla_\theta P_\theta(y)^\top}{P_\theta^3(y)} \Big]
\end{aligned}
\end{equation}
Substitute the two terms \eqref{eq:term1} and \eqref{eq:term2} into equation~\eqref{eq:hessian_p_z_y}, we obtain the Hessian of $P_\theta(z_T|y)$ (equation \eqref{eq:hessian_p_z_y_prop}). 

\subsection{Convergence analysis for the initial-state opacity case}
\label{app:proof-initial-state}
To prove the Lipschitz continuity of $H(S_0|Y;\theta)$, we need the Hessian of $H(S_0|Y;\theta)$. Take the gradient on the right side of equation~\eqref{eq:HMM_gradient_entropy_inital}, we have
\begin{equation}
\label{eq:HMM_hessian_posterior_entropy}
\begin{aligned}
  \nabla_\theta^2 H(S_0|Y;\theta) 
& =  - \sum_{y \in \mathcal{O}^T} \sum_{s_0 \in \calS} \Big[\nabla_\theta^2 P_\theta(s_0, y) \log P_\theta(s_0 | y) \\
&+ \nabla_\theta P_\theta(s_0, y) \nabla_\theta  \log P_\theta(s_0 | y)^\top 
\\
&+ P_\theta(s_0, y) \nabla_\theta^2 \log P_\theta(s_0 | y)
\\
&+ \nabla_\theta  \log P_\theta(s_0 | y) \nabla_\theta P_\theta(s_0, y)^\top \Big].
\end{aligned}
\end{equation}

\begin{lemma}
\label{lem:first_order_backward}
If $\nabla_\theta \log \pi_\theta(a \mid s)$ is bounded, then the gradient $\nabla_\theta P_\theta(s_0, y)$ is also bounded.
\end{lemma}

\begin{proof}
We employ the backward $\beta$ messages from HMM \cite{baum1970maximization}, defined as follows: 
Given a observation sequence $o_{[t:T]}$, for each $0\le t \le T$, the backward $\beta$ message at the time step $t$ for a given state $i$ is,
\begin{equation}
\label{eq:backward_definition}
\beta_t(i,\theta) \coloneqq P(o_t,o_{t+1},\dots,o_T|S_t = i),
\end{equation}
which represents the probability of having the observation sequence $o_{[t:T]}$ given the state is $i$ at time step $t$. It can be computed by the backward algorithm:
\begin{equation}
\label{eq:backward_algorithm}
\beta_t(i, \theta) = \sum_{j \in \mathcal{S}} \beta_{t+1}(j, \theta) P_\theta(i,j) b_j(o_{t+1}),
\end{equation}
for $ t=0,\ldots T-1$, and $\beta_T(i,\theta) = 1$ for $1 \le i \le N$. Also
note that $P_\theta(y|i) = \beta_0(i,\theta)$ for each state $i$ in the support of the initial state distribution. Thus, the probability $P_\theta(s_0, y) = \beta_0(s_0, \theta) \mu_0(s_0)$ for a given initial state $s_0$. 

For each $i\in \supp(\mu_0)$, we can also obtain the gradient of $P_\theta(y|i)$   by
\begin{multline}
\label{eq:gradient_backward_algorithm}
 \nabla_\theta \beta_t(i, \theta) =
\sum_{j \in \mathcal{S}} [ P_\theta(i,j) b_j(o_{t+1}) \nabla_\theta \beta_{t+1}(j, \theta) \\
 +  b_j(o_{t+1}) \beta_{t+1}(j, \theta) \nabla_\theta  P_\theta(i,j)], 
\end{multline}
and $\nabla_\theta \beta_T(i, \theta) = 0$ since $\beta_T(i, \theta)$ does not depends on $\theta$. $\nabla_\theta  P_\theta(i,j)$ is bounded from the proof of Lemma~\ref{lem:first_order}. According to the recursion~\eqref{eq:gradient_backward_algorithm}, $\nabla_\theta \beta_0(j, \theta)$ is bounded because it is a finite summation of bounded gradients. Therefore, $\nabla_\theta P_\theta(s_0, y)$ is bounded.
\end{proof}

\begin{proposition}
\label{prop:second_order_backward}
Given an observation sequence $y$ and $s_0\in \calS$ such that $P_\theta(s_0|y) \neq 0$ and $P_\theta(y) \neq 0$, the Hessian of $\log P_\theta(s_0|y)$ can be calculated as
\begin{equation}
\label{eq:hessian_log_P_s0_y}
\begin{aligned}
\nabla_\theta^2 \log P_\theta(s_0|y) = \frac{\nabla_\theta^2 P_\theta(s_0|y)}{P_\theta(s_0|y)}  - \frac{\nabla_\theta P_\theta(s_0|y) \nabla_\theta P_\theta(s_0|y)^\top}{P_\theta^2(s_0|y)} 
\end{aligned}
\end{equation}
where
\begin{equation}
\label{eq:hessian_p_s0_y_prop}
\begin{aligned}
\nabla_\theta^2 & P_\theta(s_0|y) = \sum_{s_{0} \in \calS} \mu_0(s_0) \\
& \cdot \Big[ \frac{\nabla_\theta^2 P_\theta(y|s_0)}{P_\theta(y)} - \frac{\nabla_\theta P_\theta(y|s_0) \nabla_\theta P_\theta(y)^\top}{P_\theta^2(y)} \\
& - \frac{\nabla_\theta P_\theta(y) \nabla_\theta P_\theta(y|s_0)^\top}{P_\theta^2(y)} - \frac{P_\theta(y|s_0) \nabla_\theta^2 P_\theta(y)}{P_\theta^2(y)} \\ 
& + \frac{2 P_\theta(y|s_0) \nabla_\theta P_\theta(y) \nabla_\theta P_\theta(y)^\top}{P_\theta^3(y)} \Big].
\end{aligned}
\end{equation}
\end{proposition}

\begin{proof}
The equation~\eqref{eq:hessian_log_P_s0_y} can be proved similarly as in equation~\eqref{eq:hessian_log_P_zT_y}. For equation~\eqref{eq:hessian_p_s0_y_prop} The gradient of $P_\theta (s_0|y)$ can be written as
\begin{equation}
\label{eq:gradient_bayes_rule_2}
\begin{aligned}
\nabla_\theta P_\theta(s_0|y) = \mu_0(s_0) [\frac{1}{P_\theta(y)} \nabla_\theta P_\theta(y|s_0) \\
- \frac{P_\theta(y|s_0)}{P_\theta^2(y)} \nabla_\theta P_\theta(y) ].
\end{aligned}
\end{equation}
Then the proposition can be proved similarly by replacing the $P_\theta(s_T, y)$ with $P_\theta(y|s_0)$ in equation~\eqref{eq:hessian_p_z_y_prop}. 
\end{proof}

\begin{lemma}
\label{lem:second_order_backward}
Assume that the gradient $\nabla_\theta \log \pi_\theta(a | s)$ and Hessian $\nabla_\theta^2 \log \pi_\theta(a | s)$ are bounded. Then the Hessian of probability $\nabla_\theta^2 P_\theta(s_0, y)$ is bounded.
\end{lemma}

\begin{proof}
Follow the proof of Lemma~\ref{lem:first_order}, the Hessian matrix $\nabla_\theta^2 P_\theta(y|s_0) = \nabla_\theta^2 \beta_0(j, \theta)$. Calculate the gradient of equation~\eqref{eq:gradient_backward_algorithm}, we obtain
\begin{equation}
\label{eq:update_hessian_backward_path_probability}
\begin{aligned}
&\nabla_\theta^2 \beta_t(j, \theta) = \sum_{i=1}^N  b_j(o_{t+1})\nabla_\theta \beta_{t+1}(i, \theta) \nabla_\theta P_\theta(i,j)^\top \\
& + \sum_{i=1}^N P_\theta(i,j) b_j(o_{t+1}) \nabla_\theta^2 \beta_{t+1}(i, \theta) \\
& + \sum_{i=1}^N \beta_{t+1}(i, \theta) b_j(o_{t+1}) \nabla_\theta^2 P_\theta(i,j),
\end{aligned}
\end{equation}
and $\nabla_\theta^2 \beta_T(j, \theta) = 0$ since $\nabla_\theta \beta_T(j,\theta) = 0$. $\nabla_\theta^2  P_\theta(i,j)$ is bounded from the proof of Lemma~\ref{lem:second_order}. And according to the recursion, $\nabla_\theta^2 \beta_0(j, \theta)$ is finite because it is a finite summation of bounded gradients and Hessian matrices. Therefore, $\nabla_\theta^2 P_\theta(y|s_0)$ is bounded.
\end{proof}

\textit{Theorem~\ref{thm:initial_state_convergence}:}
Under Assumption~\ref{assume:bounded-policy-gradient}, 
the entropy $H(S_0|Y; \theta)$ is Lipschitz-continuous and Lipschitz-smooth in $\theta$.

\begin{proof}
With the Lemma~\ref{lem:first_order_backward}, Proposition~\ref{prop:second_order_backward}, Lemma~\ref{lem:second_order_backward}, the theorem can be proved similarly as in the proof of theorem~\ref{thm:last_state_convergence}.
\end{proof}

\end{document}